\documentclass[sigconf, nonacm]{acmart}

\newcommand\vldbdoi{XX.XX/XXX.XX}
\newcommand\vldbpages{XXX-XXX}
\newcommand\vldbvolume{19}
\newcommand\vldbissue{4}
\newcommand\vldbyear{2025}
\newcommand\vldbauthors{\authors}
\newcommand\vldbtitle{\shorttitle} 
\newcommand\vldbavailabilityurl{URL_TO_YOUR_ARTIFACTS}
\newcommand\vldbpagestyle{empty} 

\setcopyright{none}
\AtBeginDocument{%
  \providecommand\BibTeX{{%
    \normalfont B\kern-0.5em{\scshape i\kern-0.25em b}\kern-0.8em\TeX}}}

\usepackage[scaled=.8]{beramono}
\usepackage{booktabs}
\usepackage{color}
\usepackage{colortbl}
\usepackage{tabularray}
\usepackage{float}
\AtBeginDocument{%
  \providecommand\BibTeX{{%
    \normalfont B\kern-0.5em{\scshape i\kern-0.25em b}\kern-0.8em\TeX}}}

\usepackage{makecell}
\usepackage{soul}
\usepackage{bbm}
\usepackage{amsthm}
\usepackage{balance}
\usepackage{multirow}
\usepackage[ruled,linesnumbered,vlined]{algorithm2e}
\usepackage{tcolorbox}
\usepackage[noend]{algpseudocode}
\usepackage[flushleft]{threeparttable}
\usepackage{fbox}
\usepackage{xcolor}
\usepackage{wasysym}
\usepackage{booktabs}
\usepackage{verbatim}
\usepackage{pgfplots}
\pgfplotsset{compat=1.18}
\usepackage{hyperref}
\usepackage{fontawesome}
\usepackage{tikz}
\usepackage{graphics}
\usepackage{multirow}
\usepackage{fancybox}
\usepackage{listings}
\usepackage{enumitem}
\usetikzlibrary{patterns}
\usepackage{amsmath, amsfonts}
\usepackage{graphicx}
\usepackage{subfigure}
\usepackage{textcomp}
\usepackage{xcolor}
\usepackage{cleveref}
\newtheorem{problem}{Problem}

\renewcommand\footnotetextcopyrightpermission[1]{}
\definecolor{add}{rgb}{0.2,0.4,0.6}

\begin{document}


\title{Scalable Approximate Biclique Counting over~Large~Bipartite~Graphs}




\textfloatsep 1mm plus 1mm \intextsep 1mm plus 1mm

\author{Jingbang Chen}
\authornotemark[1]
\email{chenjb@cuhk.edu.cn}
\affiliation{%
  \institution{CUHK-Shenzhen \& SLAI}
\country{}
}

\author{Weinuo Li}
\authornotemark[1]
\email{liweinuo@zju.edu.cn}
\affiliation{%
  \institution{Zhejiang University}
\country{}
}

\author{Yingli Zhou}
\authornote{The first three authors contributed equally to this research.}
\email{yinglizhou@link.cuhk.edu.cn}
\affiliation{%
  \institution{CUHK-Shenzhen}
\country{}
}

\author{Hangrui Zhou}
\email{zhouhr23@mails.tsinghua.edu.cn}
\affiliation{%
  \institution{Tsinghua University}
\country{}
}

\author{Qiuyang Mang}
\email{qiuyangmang@link.cuhk.edu.cn}
\affiliation{%
  \institution{CUHK-Shenzhen}
\country{}
}

\author{Can Wang}
\email{wcan@zju.edu.cn}
\affiliation{%
  \institution{Zhejiang University}
\country{}
}
\author{Yixiang Fang}
\email{fangyixiang@cuhk.edu.cn}
\affiliation{%
  \institution{CUHK-Shenzhen}
\country{}
}

\author{Chenhao Ma}
\authornote{Chenhao Ma is the corresponding author.}
\email{machenhao@cuhk.edu.cn}
\affiliation{%
  \institution{CUHK-Shenzhen}
\country{}
}




\captionsetup{font=small}
\setlength{\abovecaptionskip}{2.5pt}   
\setlength{\belowcaptionskip}{3pt}   
\setlength{\textfloatsep}{2pt}
\setlength{\floatsep}{2pt}

\newcommand{\todo}[1]{\textcolor{red}{[todo: #1]}}
\newcommand{\jb}[1]{\textcolor{blue}{[cjb: #1]}}
\newcommand{\yu}[1]{\textcolor{green}{[gy: #1]}}
\newcommand{\mqy}[1]{\textcolor{violet}{[mqy: #1]}}
\newcommand{\mqytext}[1]{\textcolor{violet}{#1}}
\newcommand{\hhz}[1]{\textcolor{teal}{[hhz: #1]}}
\newcommand{\mch}[1]{\textcolor{purple}{[mch: #1]}}
\newcommand{\zhou}[1]{\textcolor{magenta}{[zhou: #1]}}
\newcommand{\pr}[2]{\mathbb{P}_{#1}\left[#2\right]}
\newcommand{\expec}[2]{\mathbb{E}_{#1}\left[#2\right]}
\newcommand{\Otil}{\widetilde{O}}
\newcommand{\modify}[1]{\textcolor{blue}{#1}}
\begin{abstract}
Counting \((p,q)\)-bicliques in bipartite graphs is crucial for a variety of applications, from recommendation systems to cohesive subgraph analysis. 
Yet, it remains computationally challenging due to the combinatorial explosion to exactly count the \((p,q)\)-bicliques. 
In many scenarios, e.g., graph kernel methods, however, exact counts are not strictly required. 
To design a scalable and high-quality approximate solution, we novelly resort to \emph{\((p,q)\)-broom}, a special spanning tree of the $(p,q)$-biclique, which can be counted via graph coloring and efficient dynamic programming.
Based on the intermediate results of the dynamic programming, we propose an efficient sampling algorithm to derive the approximate $(p,q)$-biclique count from the \((p,q)\)-broom counts.
Theoretically, our method offers unbiased estimates with provable error guarantees. 
Empirically, our solution outperforms existing approximation techniques in both accuracy (up to 8$\times$ error reduction) and runtime (up to 50$\times$ speedup) on nine real-world bipartite networks, providing a scalable solution for large-scale \((p,q)\)-biclique counting.
\end{abstract}

\maketitle

\pagestyle{\vldbpagestyle}
\begingroup\small\noindent\raggedright\textbf{PVLDB Reference Format:}\\
\vldbauthors. \vldbtitle. PVLDB, \vldbvolume(\vldbissue): \vldbpages, \vldbyear.\\
\href{https://doi.org/\vldbdoi}{doi:\vldbdoi}
\endgroup
\begingroup
\renewcommand\thefootnote{}\footnote{\noindent
This work is licensed under the Creative Commons BY-NC-ND 4.0 International License. Visit \url{https://creativecommons.org/licenses/by-nc-nd/4.0/} to view a copy of this license. For any use beyond those covered by this license, obtain permission by emailing \href{mailto:info@vldb.org}{info@vldb.org}. Copyright is held by the owner/author(s). Publication rights licensed to the VLDB Endowment. \\
\raggedright Proceedings of the VLDB Endowment, Vol. \vldbvolume, No. \vldbissue\ %
ISSN 2150-8097. \\
\href{https://doi.org/\vldbdoi}{doi:\vldbdoi} \\
}\addtocounter{footnote}{-1}\endgroup

\ifdefempty{\vldbavailabilityurl}{}{
\vspace{.3cm}
\begingroup\small\noindent\raggedright\textbf{PVLDB Artifact Availability:}\\
The source code, data, and/or other artifacts have been made available at \url{https://github.com/lwn16/Biclique}.
\endgroup
}

\section{Introduction}

The bipartite graph stands as a cornerstone in graph mining, comprising two distinct sets of vertices where edges exclusively link vertices from different sets. 
This concept serves as a powerful tool for modeling relationships across various real-world domains, including recommendation networks \cite{chen2020social}, collaboration networks \cite{beutel2013copycatch}, and gene coexpression networks \cite{kaytoue2011mining}. 
For example, in recommendation networks, users and items typically constitute two distinct vertex types. The interactions between vertices form a bipartite graph, with edges representing users' historical purchase behaviors.

In the realm of bipartite graph analysis, counting $(p,q)$-bicliques has attracted much attention. A $(p,q)$-biclique is a complete subgraph between two distinct vertex sets of size $p$ and $q$. The $(2,2)$-bicliques, also known as butterflies, have played an important role in the analysis of bipartite networks~\cite{mang2024efficient,wang2014rectangle,wang2019vertex}. Furthermore, there are more scenarios in which $p,q$ is not fixed to $2$. 
In general, counting $(p,q)$-bicliques serves as a fundamental operator in many applications, including cohesive subgraph analysis~\cite{borgatti1997network} and information aggregation in graph neural networks~\cite{yang2021p}, and densest subgraph mining~\cite{mitzenmacher2015scalable}. In higher-order bipartite graph analysis, the clustering coefficient is the ratio between the counts of $(p,q)$-bicliques and $(p,q)$-wedges, where counting $(p,q)$-wedges can be reduced to counting $(p,q)$-bicliques~\cite{ye2023efficient}.

Despite its importance, counting $(p,q)$-bicliques is very challenging due to its exponential increase with respect to $p$ and $q$~\cite{yang2021p}. For example, in one of the graph datasets \texttt{Twitter}, with fewer than $2 \times 10^6$ edges, there are more than $10^{13}$ $(5,4)$-bicliques and more than $10^{18}$ $(6,3)$-bicliques within. Therefore, enumeration-based counting methods including \texttt{BCList++}~\cite{yang2021p}, \texttt{EPivoter}~\cite{ye2023efficient} that produce the exact solution are not scalable. On the other hand, in many applications of biclique counting, an approximate count is often sufficient. In \textbf{graph kernel methods}~\cite{sheshbolouki2022sgrapp}, motifs serve as the basis for defining similarity measures between graphs in tasks like classification and anomaly detection. Since graph kernels depend on relative similarities rather than exact counts, approximate counts can preserve kernel performance while significantly reducing the computational overhead. This efficiency facilitates similarity computations in various fields:

\begin{itemize}[leftmargin=*]
    \item \textbf{Bioinformatics}: Zhang et al. \cite{zhang2014finding} proposed algorithms to enumerate maximal bicliques in bipartite graphs derived from functional genomics (genes vs. gene-sets), enabling the identification of co-regulatory modules. A high number of bicliques in certain regions often indicates dense regulatory relationships, which can be helpful for detecting such modules. However, full enumeration quickly becomes infeasible as graphs grow larger or the number of bicliques explodes. This limitation motivates our work on scalable approximate biclique counting, which aims to estimate biclique numbers (or those of particular sizes) accurately without exhaustive enumeration.
    \item \textbf{Cybersecurity}: In e-commerce fraud detection, bicliques also serve as key indicators of suspicious activity. For example, Wu et al. \cite{wu2024accelerating} show that in Taobao’s large transaction network, click-farming attacks can be modeled as large bicliques. Similarly, Ban \cite{ban2018finding} introduces the Maximal Half-Isolated Biclique problem to detect synchronized fraudulent behavior in customer–product graphs. These studies highlight that high biclique counts—or the presence of unusually large bicliques—in specific regions of a bipartite graph can be strong signals of fraud.
    \item \textbf{Data Mining}: In recommendation systems~\cite{wang2014rectangle}, biclique counting helps uncover dense substructures among users or items, such as groups of users with shared interests or items commonly co-purchased. In large-scale environments such as e-commerce or streaming platforms, allowing a small error margin for counting can significantly reduce the runtime cost while still preserving the effectiveness of the downstream tasks.
\end{itemize}
From this perspective, we can see that developing algorithms that produce approximate answers to $(p,q)$-biclique counting is more practical and more applicable to large-scale data. Ye et al. propose a sampling-based method called \texttt{EP/Zz++} for approximate counting $(p,q)$-bicliques~\cite{ye2023efficient}. However, its precision is not satisfactory as $p,q$ increases. For example, when querying $(5,9)$-bicliques in the data set \texttt{Twitter}, \texttt{EP/Zz++} could produce an answer of the error ratio of more than $200\%$.\footnote{Let $C$ and $\hat{C}$ be the exact and approximate counts of the $(p,q)$-biclique, respectively. The estimation error ratio is defined as $\frac{|\hat{C}-C|}{C}$.}
This shows the need to develop a better algorithm for approximate $(p,q)$-biclique counting, aiming to improve scalability and accuracy.

In this paper, we propose a new sampling-based method that produces a high-accuracy approximate counting of $(p,q)$-bicliques. We first adapt the coloring trick in this problem, inspired by several counting methods of $k$-cliques in general graphs~\cite{li2020ordering,ye2022lightning}. Then, we design a special type of subgraph that has a strong correlation with $(p,q)$-biclique, named $(p,q)$-brooms. The $(p,q)$-broom is a special type of spanning tree in a $(p,q)$-biclique, which has a fixed structure that can be taken advantage of when counting the amount.
Using dynamic programming, we can count this motif efficiently and precisely. It also naturally adapts to the coloring.  
Finally, we develop a sampling method that takes advantage of the coloring and the intermediate result of dynamic programming. 
Compared to the h-zigzag pattern proposed by \texttt{Ep/Zz++}, our $(p,q)$-brooms have a structure that is closer to \((p,q)\)-cliques, and they usually have a smaller amount in the graph. As a result, it provides a better error guarantee when sampling \((p,q)\)-bicliques from them. Our algorithm is still heuristic since it includes a graph-coloring process, and the sampling time analysis contains a ratio that could be exponential. However, these processes are shown to be very efficient in practice. We empirically evaluate our algorithm against state-of-the-art exact and approximate counting algorithms on nine real-world datasets. All experimental results demonstrate the effectiveness of our algorithm. Besides the mathematical analysis, the high-level idea here is the sparsification. Since $(p,q)$-bicliques in real-world datasets contain realistic meanings, finding a sparsifier that is likely to preserve the information within bicliques is a promising approach to reducing the computational cost in practice. On the other hand, the mathematical analysis of unbiasedness and sampling time ensures that we can use this algorithm in practice. We can set the sampling time to a large number, allowing us to withstand the overall cost. Or we can iteratively increase the sampling time and terminate while the result shows a convergence trend.

Our contributions are summarized as follows:
\begin{itemize}[leftmargin=*]
    \item We design $(p,q)$-broom, a special type of spanning tree in a $(p, q)$-biclique, which can be efficiently counted via dynamic programming.
    \item Based on the dynamic programming result of counting $(p,q)$-brooms, we develop an efficient sampling-based algorithm that computes a high-accuracy counting of $(p,q)$-bicliques in bipartite graphs.
    \item Mathematical analysis shows that our algorithm is unbiased and obtains a provable error guarantee.
    \item Extensive experiments show that our algorithm consistently outperforms state-of-the-art approximate algorithms, with up to 8$\times$ reduction in approximation errors and up to $50 \times$ speed-up in running time.
    \item Additional experiments further support the effectiveness and novelty of our algorithm, which includes empirical results such as an ablation study on the newly designed coloring method and an experiment on the hyper-parameters.
\end{itemize}

\noindent\textbf{Outline.} The rest of the paper is organized as follows. We review the related work in \Cref{sec:related}, and introduce notations and definitions in \Cref{sec:preliminaries}. \Cref{sec:algo} presents our sampling-based algorithm. Experimental results are given in \Cref{sec:exp} and \Cref{exp:detail}. We present major experimental results in \Cref{sec:exp} and provide several detailed analyses. We conclude in \Cref{sec:conc}.
\section{Related Works}
\label{sec:related}
\textit{Biclique Counting.} The biclique counting problem focuses on enumerating $(p, q)$-bicliques within bipartite graphs, with particular emphasis on the $(2, 2)$-biclique, commonly known as a butterfly—a fundamental structural pattern with significant real-world applications. 
Given the importance of butterfly motifs, researchers have developed numerous algorithms for their efficient enumeration and counting~\cite{wang2019vertex,sanei2018butterfly, wang2014rectangle}.
The state-of-the-art method utilizes the vertex priority and cache optimization~\cite{wang2019vertex}. 
Recent advances have emerged in multiple directions: parallel computing techniques that optimize both memory and time efficiency~\cite{xia2024gpu}, I/O-efficient methods~\cite{wang2023efficient} that minimize disk operations, and approximation strategies that provide accurate counts while reducing computational overhead. 
The field has further expanded to address more complex graph types, including uncertain graphs with probabilistic edges~\cite{zhou2021butterfly} and temporal graphs incorporating time-varying relationships~\cite{cai2024efficient}.
The $(p,q)$-biclique counting is a general version of butterfly counting, which recently gained much attention~\cite{ye2023efficient,yang2021p}. 
The state-of-the-art methods are introduced and analyzed in the latter (See Section 3.1). 

\textit{$k$-clique Counting.} The evolution of $k$-clique counting algorithms began with Chiba and Nishizeki's backtracking enumeration method, which was later enhanced through more efficient ordering-based optimizations, including degeneracy ordering (Danisch et al.~\cite{eden2020faster}) and color ordering (Li et al.~\cite{li2020ordering}). While these algorithms perform efficiently for small $k$, their performance deteriorates with increasing $k$.
\texttt{Pivoter}~\cite{jain2020power}, introduced by Jain and Seshadhri, marked a significant advancement by employing pivoting techniques from maximal clique enumeration, enabling combinatorial counting instead of explicit enumeration. 
However, \texttt{Pivoter}'s performance can degrade on large, dense graphs. 
To address scalability challenges, researchers have developed sampling-based approaches, including \texttt{TuranShadow}~\cite{jain2020provably} and coloring-based sampling techniques~\cite{ye2022lightning}.

\textit{Other Motif Counting.} Beyond $k$-clique and $(p,q)$-biclique, numerous works are focusing on counting general motifs, such as four-cycles and other subgraph patterns~\cite{hou2024learnsc,zhao2023learned,bressan2018motif,ma2019linc,seshadhri2019scalable,zhou2024efficient}. These methods can be broadly divided into two categories: traditional counting approaches~\cite{bressan2018motif,ma2019linc,seshadhri2019scalable} and learning-based methods~\cite{hou2024learnsc,zhao2023learned}. 
The former is typically based on the sampling and enumeration techniques used to obtain the approximate and exact count of the motif, respectively. 
In contrast, learning-based approaches offer an alternative approximate solution by leveraging machine learning techniques. Beyond this primary categorization, subgraph counting can be further classified by scope: global counting determines subgraph frequencies across the entire graph, while local counting focuses on specific nodes or edges.

\textit{Dynamic Programming in Subgraph Problems} In our proposed algorithm, dynamic programming is used to compute the number of special types of spanning trees. In many recent works for
subgraph matching algorithms~\cite{kim2021versatile,han2019efficient}, they also use dynamic programming to search a data
graph for embeddings of a path tree obtained from a
query graph. Although we all fix an ordering for these trees and conduct dynamic programming according, our method focuses and works differently: We use both the intermediate and final results to ensure the following sampling process can be effective. Furthermore, our transition formula and recurrence relations are significantly different. The key idea of our method is to take advantage of the bipartite structure. We split the tree into several "prefix" subtrees and identify them by the number of edges and the last edge that adds in. In this way, we manage to compute the number of all trees as well as the number of all "prefix" tree.

\vspace{-1ex}
\section{Preliminaries}
\label{sec:preliminaries}
Throughout the paper, we study the bipartite graphs. A bipartite graph is an undirected graph $G = ((U,V), E)$, where $U$ and $V$ are disjoint sets of vertices and each edge $(u,v) \in E$ satisfies $u \in U, v \in V$. To denote the corresponding vertex and edge sets of a certain $G$, we may use $U(G)$, $V(G)$, and $E(G)$, respectively. For each vertex $u \in U(G)$, we denote its neighbor set as $N(u,G) = \{v \in V \mid (u,v) \in E \}$. Similarly, for a vertex $v \in V(G)$, its neighbor set is $N(v,G) = \{u \in U \mid (u,v) \in E \}$. For any vertex $x \in U \bigcup  V$, its degree $d(x)$ is the size of its neighbor set. Now we define the $(p,q)$-biclique.
\begin{definition}[Biclique]
    Given a bipartite graph $G = ((U,V),E)$, a $(p,q)$-biclique is a subgraph $G' = ((U',V'),E') \in G$ where $|U'|=p$, $|V'|=q$ and $E'=\{(u,v)\mid u\in U', v\in V'\}.$
\end{definition}

\begin{figure}[t!]
    \centering
    \includegraphics[width=0.7\linewidth]{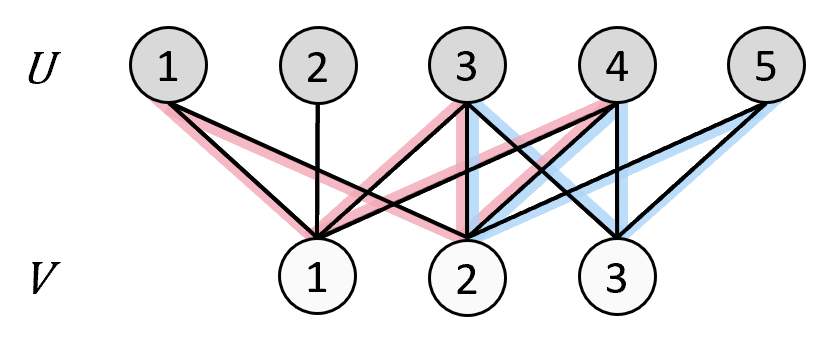}
    \caption{An illustrative example of two $(3,2)$-bicliques.}
    \label{fig:biclique_counting}
\end{figure}

Lastly, we define the problem that we mainly study:
\begin{tcolorbox}[boxsep=1pt,left=2pt,right=2pt,top=3pt,bottom=2pt,width=\linewidth,colback=white,boxrule=0.6pt, colbacktitle=white!,toptitle=2pt,bottomtitle=1pt,opacitybacktitle=0]
\begin{problem}[Biclique Counting]
    \label{prob:biclique-counting}
    Given a bipartite graph $G$ and two integers $p$ and $q$, compute the number of $(p,q)$-bicliques in $G$.
\end{problem}
\end{tcolorbox}
For example, in \Cref{fig:biclique_counting}, there are two $(3,2)$-bicliques: \textbf{(1)} $U'=\{1,3,4\},V'=\{1,2\}$; \textbf{(2)} $U'=\{3,4,5\},V'=\{2,3\}$. Their edges are highlighted with pink lines and blue lines, respectively. Therefore, the number of $(3,2)$-bicliques is $2$. It is also known that \Cref{prob:biclique-counting} is NP-Hard~\cite{yang2021p}, which means there is no algorithm that runs in polynomial time. Throughout this paper, we assume $p,q \geq 2$.

\newcommand{\kw}[1]{{\ensuremath{\texttt{#1}}}}

\section{ALGORITHM}
\label{sec:algo}
In this section, we propose a new algorithm called \textit{\underline{C}olored \underline{B}room-based \underline{S}ampling (\texttt{CBS})} that solves \Cref{prob:biclique-counting} approximately.

\paragraph{1. Coloring.} Inspired by methods for solving $k$-clique counting problems on general graphs \cite{li2020ordering,ye2022lightning}, we start by assigning a color to every vertex on the graph (\Cref{sec:coloring}). Our coloring method is simple, and produces assignments with a small number of colors empirically. Although graph coloring has already been used in k-clique problems on general graphs, it is a new attempt to utilize it on bipartite graphs for biclique problems. On the other hand, if we directly apply the traditional coloring, it becomes 2-coloring and cannot ensure the distinct color assignment for every existing biclique. In general, coloring can be considered as an optimization aiming to reduce the number of accountable brooms. Intuitively, each $(p,q)$-biclique must obtain at least $p$ or $q$ colors due to the definition. In turn, increase the probability such that the sampled broom is a biclique. It also allows us to construct a dynamic programming implemented by following the color indices.

We construct a new coloring condition and prove its effectiveness in result. Since the graph-coloring problem is NP-Complete, there should be no theoretically optimal guarantee. In our empirical study, we have conducted an ablation study (\Cref{exp:color}) showing the effectiveness of the coloring process as well as slightly different coloring strategies. 

\paragraph{2. Broom Pattern Counting.} Instead of counting bicliques directly, our next step is to calculate the number of a special motif called "broom" that we design (\Cref{sec:patterncounting}). It is a special type of spanning trees of $(p,q)$-bicliques. Our first insight is that a spanning tree captures a subgraph better than a simple path. Then, since in the real-world dataset, $(p,q)$-bicliques contain realistic meanings and we also hope the sparsified motif can preserve this inner information. Therefore, we try to allocate all $p+q-1$ edges to each vertex as evenly as possible. In turn, we have the following mathematical definition of $(p,q)$-brooms that indicates this purpose: 
\begin{definition}[Broom]
\label{def:broom}
Given a bipartite graph $G((U,V),E)$, two vertex sets $U' \in U$ and $V' \in V$ of size $p$ and $q$, respectively. A vertex ordering is also given: $\{u_1,u_2,\dots,u_{p}\}$ and $\{v_1,v_2,\dots,v_{q}\}$. The corresponding $(p,q)$-broom is the subgraph $G'((U',V'),E') \in G$ that satisfies 
\begin{align*}
E'=\{( u_i,v_{\lfloor \frac{(i-1)(q-1)}{p-1} + 1\rfloor} )\mid 1\leq i \leq p\} \\ \cup \{(u_{\lceil \frac{(i-1)(p-1)}{q-1}\rceil},v_i) \mid 2\leq i \leq q\}.
\end{align*}
\end{definition}


\begin{figure}[t!]
    \centering
    \includegraphics[width=0.8\linewidth]{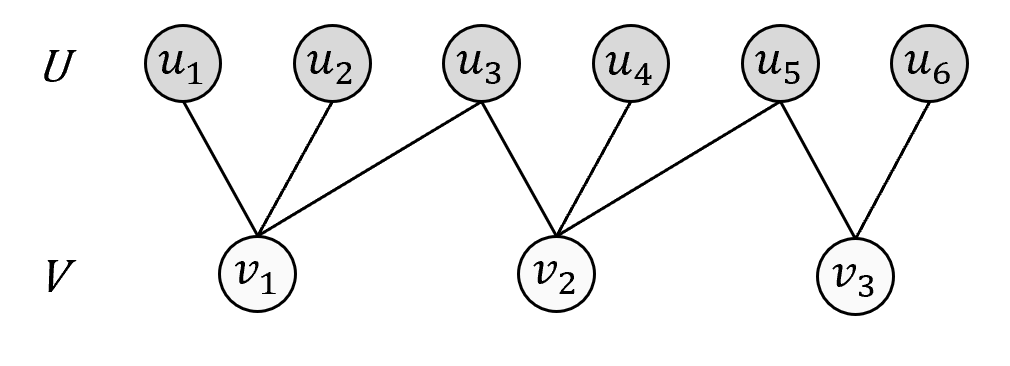}
    \caption{An illustrative example of a $(6,3)$-broom.}
    \label{fig:broom}
\end{figure}

In \Cref{fig:broom}, we provide an illustration of a $(6,3)$-broom. We can observe that such a subgraph is, in fact, a tree, which contains $p+q-1$ edges and is connected. As shown in \Cref{fig:broom}, we can clearly see that each vertex $v_i$ in $V$ connects to the $2-3$ vertices $u_j$ in $U$ that $i$ and $j$ are close. Note that the defined ordering is only used to characterize the pattern, providing a better visualization, and it does not have to align with the original vertex index. Visually speaking, a $(p,q)$-broom is the "skeleton" of a $(p,q)$-biclique. The general idea here is to design a way to sparsify the dense biclique while trying best to preserve its identity. This broom pattern is very different from \texttt{EP/Zz++}~\cite{ye2023efficient}'s $h$-zigzag patterns, as they only use a simple path. We then empirically show that such pattern's count can lead to a much better approximation of biclique counting. 
This distribution design limits the total number of existing brooms, which makes the error analysis tight. Despite its complicated structures, we show in \Cref{sec:patterncounting} that its amount can be computed precisely and efficiently using a dynamic programming process with respect to color assignment.

\paragraph{3. Counting by Sampling Brooms.} Lastly, by sampling, we extract the quantity relationship between the $(p,q)$-brooms and $(p,q)$-bicliques in the graph and eventually compute our approximate answer (\Cref{sec:sampling}). Shown in \Cref{sec:exp}, our algorithm produces high-accuracy approximate counting with faster runtime in real-world datasets. We also prove unbiasedness and provide an error guarantee for our algorithm.


\subsection{Coloring}
\label{sec:coloring}
The color assignment should guarantee that for any $(p,q)$-clique $H$, there are no two vertices in $U(H)$ that share the same color. Similarly, the same is true for any two vertices in $V(H)$. The way to achieve this is to relax the condition as "limiting same-color-endpoint-wedges $\leq q$", which is easy to keep track of.

Our coloring algorithm is shown in \Cref{alg:Coloring}. We run the coloring process \kw{CalcColor} for the vertices in $U$ and $V$ separately (\Cref{Coloring:calc_Col_Left}, \ref{Coloring:calc_Col_Right}). Note that here, we consider the neighbor symmetrically. After preprocessing the input graph $G$ with \Cref{alg:Coloring}, each vertex $x \in U\bigcup V$ is colored, and we denote its color as $c(x)$. 

$S$ denotes the set of vertices that has not been colored. Whenever $S$ is not empty (\Cref{Coloring:notempty}), we will try coloring some vertices with a new color, denoted as $\kappa$. For brevity, we initialize $\kappa$ to be $1$ (\Cref{Coloring:kappa_init}), and whenever we need a new color, we increase $\kappa$ by $1$ (\Cref{Coloring:new_color}). $Cnt$ is a temporary array that is used to guarantee the coloring is legal. Specifically, when we try coloring with color $\kappa$, $Cnt[u]$ stores the number of $v$ that satisfies $v \in N(u,G)$ and there exists $w \in N(v,G)$ such that $w \neq u$ and $c(w) = \kappa$. In each round of coloring, we start by initializing $Cnt$ to be all zeros (\Cref{Coloring:cnt_init}). We use $S_{next}$ to store vertices that have to be colored in the next round, initialized as $\emptyset$ (\Cref{Coloring:snxt_init}). We enumerate vertices in $S$ in a randomized order (\Cref{CalcColor:TryColor}). For each enumerated vertex $u$, we first set its color to be $\kappa$ (\Cref{CalcColor:setcolor}) and update $Cnt$ correspondingly (\Cref{CalcColor:updatecnt}). 
If there exists a vertex $x$ such that $Cnt[x]\geq q$ and $c(x) = \kappa$ (\Cref{CalcColor:bad}), then there could exist a $(p,q)$-biclique with two vertices with the same color. Therefore, such coloring is illegal, and we cannot color $u$ with $\kappa$. In such case, we first revert the changes in $Cnt$ (\Cref{CalcColor:revert}) and add $u$ into $S_{next}$ (\Cref{CalcColor:addnext}). In the end, we have tried every vertex in $S$ at least once and have deferred some vertices' coloring to upcoming rounds. We assign $S_{next}$ to $S$ and continue (\Cref{CalcColor:updates}).

To see that the runtime of this coloring process is dominated by the total time of the total number of colors, i,e. $\kappa$, we need to show that querying, updating, and reverting $cnt$ can be done in $O(|N(u,G)|)$ each operation. Here we provide one possible implementation about how to achieve $|N(u,G)|$ for a single update as follows: The idea is to use the lazy propagation trick. We only maintain $cnt$ from vertices that are (will be) colored $\kappa$, and each round we maintain a status $status[v] = (t, w)$ for $V$ during each round coloring $U$. Specifically, $t=0$ indicates $v$ has no neighbor of color $\kappa$; $t=1$ means exactly one neighbor $w$ of color $\kappa$, and $t=2$ means at least two neighbors are colored $\kappa$. Whenever we color $u$ to $\kappa$, we need to check all its neighbor $v \in N(u,G)$. Let $status[v]=(t,w)$. If $t\geq 1$,we update $Cnt[u]$. If $t=1$, we update $Cnt[w]$. After it, we update $status[v]$ accordingly. In all, for each $u$, we update its neighbor $v$ with $O(1)$ operations. Therefore, the total cost for a single update is $O(|N(u,G)|)$. In all, the time complexity is $O(\kappa|E|)$. Note that $\kappa$ is naturally bounded by $|U(H)|$. In fact, shown in \Cref{tab:dataset}, on real-world datasets, $\kappa$ is not large and our coloring algorithm runs very efficiently in turn. The only question left is how such an algorithm guarantees a correct color assignment. We prove it in the following lemma.

\begin{lemma}
After executing \kw{Coloring}$(G,p,q)$, for any $(p,q)$-clique $H$ in $G$, there are no two vertices $x,y$ either both in $U(H)$ or in $V(H)$ that share the same color, i.e. $c(x)=c(y)$.
\end{lemma}

\begin{proof}
Let $U(H)=\{u_1,u_2,\dots,u_{p}\}$, $V(H)=\{v_1,v_2,\dots,v_{q}\}$. Because of symmetry, we only prove the case of $U$ and $p$. We assume that there are two vertices $u_1,u_2 \in U(H)$ with the same color, i.e. $c(u_1)=c(u_2)=k$. WLOG, when we color vertices with the $k$-th color, we first set $c(u_1) \leftarrow k$. Then we will have $Cnt[u_2] \ge q$ since $\{v_1,v_2,\dots,v_{q}\}$ are $u_2$'s neighbors and they all have neighbor $u_1$ such that $c(u_1)=k$. So $u_2$ cannot be colored in this round. This contradicts the assumption.

\end{proof}




\begin{figure}[t!]
    \centering
    \includegraphics[width=0.75\linewidth]{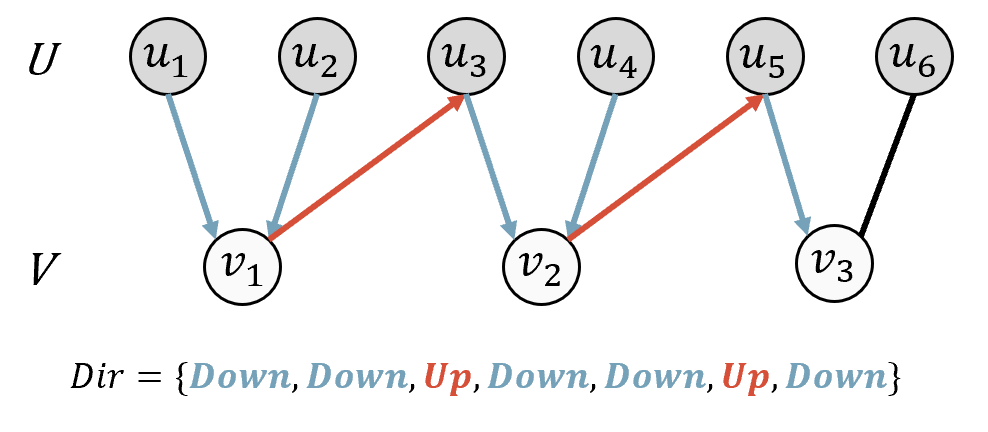}
    \caption{An example of edges' directions in a $(6, 3)$-broom, where edges with different directions are represented by $\texttt{Up}$ and $\texttt{Down}$ in $Dir$ and colored in red and blue, respectively.}
    \label{fig:broom_direction}
\end{figure}
\begin{algorithm}[t]
\small
\caption{\texttt{CBS: Coloring}}\label{alg:Coloring}
\KwIn{$G$: a bipartite graph; $p,q$: two parameters.}
\KwOut{$G'$: a colored bipartite graph.}
\SetKwProg{Fn}{Procedure}{}{}
\SetNoFillComment

\kw{CalcColor}$(G,U(G),p,q)$;\\ \label{Coloring:calc_Col_Left}
\kw{CalcColor}$(G,V(G),q,p)$;\\ \label{Coloring:calc_Col_Right}
$G' \leftarrow G$;\\

\Return $G'$; \label{Coloring:return}

\Fn{\kw{CalcColor}$(G,S,p,q)$}{ 
    Initialize $\kappa \leftarrow 1$; \\ \label{Coloring:kappa_init}
    \While{$|S| > 0$}{  \label{Coloring:notempty}
        Initialize an array $Cnt$ with zeros; \\  \label{Coloring:cnt_init}
        $S_{next} \leftarrow \emptyset$; \\ \label{Coloring:snxt_init}
        \tcc{enumerate vertices in a random order} 
        \For{$u \in S$}{ \label{CalcColor:TryColor} 
            $c(u) \leftarrow \kappa$; \\ \label{CalcColor:setcolor} 
            Update $Cnt$; \\ \label{CalcColor:updatecnt} 
            \If {$\exists x, Cnt[x]\ge q$ and $c(x) = \kappa$}{ \label{CalcColor:bad} 
                Revert $Cnt$'s change; \\ \label{CalcColor:revert} 
                $S_{next} \leftarrow S_{next} \cup \{u\}$; \\ \label{CalcColor:addnext} 
            }
        }
        $\kappa \leftarrow \kappa + 1$; \\  \label{Coloring:new_color}
        $S \leftarrow S_{next}$; \\ \label{CalcColor:updates} 
    }
}
\end{algorithm}

Randomization is a classical way to produce a solution in graph coloring. We have conducted an additional experiment in \Cref{exp:color} using three different ordering strategies: ascending order, descending order, and random order. The results indicate that there are no significant differences in efficiency or error between the three strategies. There is also another ablation study on the effectiveness of adapting the coloring process in counting bicliques. As shown in \Cref{exp:color}, with a slight increase in runtime, the coloring process enhances the accuracy in turn.
\subsection{Pattern Counting}
\label{sec:patterncounting}
After coloring, we enhance the motif we are counting with a condition of distinct colors. However, it is still hard to count bicliques directly. Instead, we count the number of $(p,q)$-brooms. Since we assign a concrete integer for each color in \Cref{alg:Coloring} for any $(p,q)$-biclique, to avoid overcounting, we assume its vertex ordering of $U'$ and $V'$ is following the color from small to large.
With this exact ordering, by \Cref{def:broom}, there is exactly one $(p,q)$-broom subgraph in each $(p,q)$-biclique. Note that the color of these $(p,q)$-brooms is also unique, which is convenient for counting. \Cref{alg:CountingIndex} is able to compute the number of $(p,q)$-brooms of this type. We then show how to use this value to approximate the biclique count.

The whole process is a dynamic programming approach, where its state is represented by the total number of edges accumulated so far and the last edge being tracked. To well define this dynamic programming state, we need to specify a few more properties in this type of $(p,q)$-broom:

\vspace{3px}
\noindent\textbf{A.} Any $(p,q)$-broom $H$ is connected and contains exactly $p+q-1$ edges. That is to say, it is actually a tree.

\vspace{3px}
\noindent\textbf{B.} If we sort all edges $(u,v) \in E(H)$ by the increasing order of the tuple $(c(u),c(v))$, there is only three types of edges: 
\begin{enumerate}
    \item The edge with the maximum tuple.
    \item Edges sharing an endpoint in $U(H)$ with the next edge.
    \item Edges sharing an endpoint in $V(H)$ with the next edge.
\end{enumerate}
Here, for sorting, the comparison between two tuples $(c(u),c(v))$ and $(c(u'),c(v'))$ is first done by verifying if $c(u)<c(u')).$ If $c(u)=c(u')$, we then check if $c(v)<c(v')$. Therefore, the maximum tuple refers to the last element after sorting in this manner. If For \textbf{(2)} and \textbf{(3)}, we signal them with a \textbf{"direction"} as either \texttt{Up} or \texttt{Down}. 
For example, as shown in the $(6,3)$-broom in \Cref{fig:broom_direction}, $(u_1,v_1)$ is with the smallest partial order. Its direction is \texttt{Down} since it shares the same endpoint $v_1$ with the next edge $(u_2,v_1)$. Similarly, $(u_2,v_1)$, $(u_3,v_2)$, $(u_4,v_2)$, $(u_5,v_3)$ are also with direction \texttt{Down}, highlighted with blue arrows. $(u_3,v_1)$, $(u_5,v_2)$ are with direction \texttt{Up}, highlighted with reversed red arrows. $(u_6,v_3)$ has no direction since it is the edge with the maximum tuple.

After fixing the structure and order of edges by \textbf{B}, we can simply use the number of edges and the last edge to represent the current structure of the growing broom. The rest is enumeration and direct transition. The details are as follows.

To begin with, we sort $U(G)$, $V(G)$ by the increasing order of colors $c(u), c(v)$ (\Cref{index:sortuv}), and sort $E(G)$ by the increasing order of the color tuple $(c(u),c(v))$ (\Cref{index:sorte}). 

Let $Dp[len][pre]$ denote the number of possible ways to build a growing broom with the current number of edges $len$ and the last edge $pre$, which follows the fixed structure and order of edges by \textbf{B}. The objective is to correctly computed all states for $1\leq len \leq p+q-1$ (there are $p+q-1$ edges in total besides the first one) and $pre \in E(G)$. We use a 2D table to store the results, initialized to $0$ (\Cref{CountingIndex:init_Dp}). The base case contains the case where each $(u,v)$ acts as the first edge in the broom: we set $Dp[1][(u,v)]$ to be $1$ (\Cref{index:init1}, \ref{index:init2}). We iterate this dimension from small to large to accumulate the growing broom from $1..t$-th edges to $1..t+1$-th edges (\Cref{index:for}).

We use $Dir[t]$ to denote the $t$-th edge's direction in the $(p,q)$-broom. When $Dir[t]$ is \texttt{Up} (\Cref{index:ifleft}), the $t+1$-th edge and the $t$-th edge share a common point in $U$. Therefore, for each edge $(u,v)$ (\Cref{index:left1}), $Dp[t+1][(u,v)]$ should accumulate the number of ways coming from any $Dp[t][(u,w)]$ such that $w$ is in $N(u,G)$ and $c(w)<c(v)$ (\Cref{index:left2}). Similarly, when $Dir[t]$ is \texttt{Down} (\Cref{index:ifright}), the $t+1$-th edge and the $t$-th edge share a common point in $V$. In this case, $Dp[t+1][(u,v)]$ should accumulate the number of ways coming from any $Dp[t][(w,v)]$ such that $w$ is in $N(v,G)$ and $c(w)<c(u)$ (\Cref{index:right1}, \ref{index:right2}). After all computations, we return $B$ as the sum of all $Dp[p+q-1][\cdot]$ and the $Dp$ table as the results (\Cref{index:return}).

By using the prefix sum technique, for each iteration, we can finish all transition computations in $O(|E|)$. There are $O(p+q)$ rounds. Therefore, the total time complexity is $O((p+q)|E|)$.

\begin{algorithm}[t]
\small
\caption{\texttt{CBS: CountingIndex}}\label{alg:CountingIndex}
\KwIn{$G$: a colored bipartite graph; $p,q$: two parameters.}
\KwOut{$Dp$: a $2D$ array; $B$: the total number of $(p,q)$-brooms in $G$.}

Sort $U(G),V(G)$ by the increasing order of $c(u),c(v)$;\\ \label{index:sortuv}
Sort $E(G)$ by the increasing order of $(c(u),c(v))$;\\ \label{index:sorte}

Initialize $Dp$ with zeros; \\ \label{CountingIndex:init_Dp}

\lFor {$(u,v) \in E(G)$}{ \label{index:init1}
    $Dp[1][(u,v)] \leftarrow 1$ \label{index:init2}
}
\For {$t\leftarrow 1$ \KwTo $p+q-2$}{ \label{index:for}
    \uIf {$Dir[t] = \texttt{Up}$}{ \label{index:ifleft}
        \For {$(u,v) \in E(G)$}{ \label{index:left1}
            $Dp[t+1][(u,v)] \leftarrow \sum_{w\in N(u,G),c(w)<c(v)} Dp[t][(u,w)]$; \\ \label{index:left2}
        } 
    } 
    \Else{ \label{index:ifright}
        \For {$(u,v) \in E(G)$}{ \label{index:right1}
            $Dp[t+1][(u,v)] \leftarrow \sum_{w\in N(v,G),c(w)<c(u)} Dp[t][(w,v)]$; \\ \label{index:right2}
        }
    }
}
$B \leftarrow \sum_{(u,v) \in E(G)} Dp[p+q-1][(u,v)]$; \\ \label{index:return}

\Return $Dp,B$; \label{CountingIndex:return}
\end{algorithm}

\subsection{Approximate Counting via Sampling}
\label{sec:sampling}
After \kw{CountingIndex}$(G,p,q)$, we have accumulated $B$ $(p,q)$-brooms. To approximate the quantity relation between our computed $(p,q)$-brooms and $(p,q)$-bicliques, we can do the following sampling:
\begin{enumerate}[leftmargin=*]
    \item Set two counters $cnt_{broom}$ and $cnt_{biclique}$.
    \item Uniformly sample a $(p,q)$-broom $P$ from all $B$ ones and increase $cnt_{broom}$ by $1$.
    \item If the induced subgraph by $U(P)$ and $V(P)$ is a $(p,q)$-biclique, increase $cnt_{biclique}$ by $1$.
\end{enumerate}

After applying this sampling process sufficiently many times, we can roughly approximate the number of $(p,q)$-bicliques by $cnt_{biclique}/cnt_{broom}\times B.$
However, such a sampling process is too inefficient. We now propose \Cref{alg:Sampling} to accelerate this process, utilizing the computed dynamic programming results from \kw{CountingIndex}$(G,p,q)$. Recall that $Dp[t][(u,v)]$ represents the number of ways to build a growing $(p,q)$-broom til the $t$-th edge, which is $(u,v)$. 
To begin with, we start by sampling the last edge of the $(p,q)$-broom, denoted as $(u_{last},v_{last})$ (\Cref{sampling:samplelast}). Here, the sampling should follow the weight distribution of $\{Dp[p+q-1][\cdot]\}$, indicating how many $(p,q)$-brooms end with each edge. The following process is similar to reverting the dynamic programming process. Instead of building the $(p,q)$-broom from the $1$-st edge to the $(p+q-1)$-th edge, we do it reversely. We use $(u_{cur},v_{cur})$, $U'$, $V'$ to denote the current growing $(p,q)$-broom. They are initialized as $(u_{last},v_{last})$, $\{u_{last}\}$, $\{v_{last}\}$ respectively (\Cref{sampling:initcur}, \ref{sampling:inituv}). Specifically, since $(u_{last},v_{last})$ is the $(p+q-1)$-th edge, we are now adding the $(p+q-2)$-th edge til the $1$-th edge gradually (\Cref{sampling:for}). 

Assume we are processing the $i$-th edge now, and we know the $(i+1)$-th edge is $(u_{cur},v_{cur})$. We initialize an edge set $S$ to store the candidate edge for the $i$-th edge. Based on the direction signaling, we can know whether it should share the common endpoint in $U$ or $V$. If $Dir[i]$ is \texttt{Up} (\Cref{sampling:ifleft}), then we should find all $w \in N(u_{cur},G)$ such that $c(w)<c(v_{cur})$ and $U' \subseteq N(w,G)$. We assign $S$ to be $\{(u_{cur},w)\}$ (\Cref{sampling:leftupdate}). The first condition is for the color order, and the second condition indicates that adding $w$ can guarantee that the induced subgraph is biclique. Similarly, if $Dir[i]$ is \texttt{Down} (\Cref{sampling:ifright}), we find all all $w \in N(v_{cur},G)$ such that $c(w)<c(u_{cur})$ and $V' \subseteq N(w,G)$. Then we assign $S$ to be $\{(w,v_{cur})\}$ (\Cref{sampling:rightupdate}).  

The idea to accelerate the sampling is, by building the $(p,q)$-broom, we strictly guarantee that it will correspond to a $(p,q)$-biclique while keeping track of the probability of sampling out this $(p,q)$-broom, which can be computed through the dynamic programming table. In the beginning, we initialize $ans$ to be $1$ (\Cref{sampling:initans}). Then, for the $i$-th edge, the probability contribution of sampling it out from $S$ should be the number of the growing broom ends with $(u,v)\in S$ divided by the total number of possible brooms at this step, which is 
$\frac{\sum_{(u,v) \in S} Dp[i][(u,v)]}{Dp[i+1][(u_{cur},v_{cur})]}.$
We multiply $ans$ by this value (\Cref{sampling:updateans}). If $ans$ becomes $0$, then there is no possible $(p,q)$-biclique from the current sampled $(p,q)$-broom. We return with $0$ in this case (\Cref{sampling:break}). Now we are ready to sample the $i$-th edge from $S$ (\Cref{sampling:samplenext}). We do so following the normalized distribution of possible brooms from this step ($\{Dp[i][\cdot]\}$). To continue the process, we assign $(u_{cur},v_{cur})$, $U'$, $V'$ to $(u_{next},v_{next})$, $U' \cup \{ u_{next}\}$, $V' \cup \{ v_{next}\}$, respectively (\Cref{sampling:updatecur}, \ref{sampling:updateuv}).

In the end, we successfully sample a $(p,q)$-broom that corresponds to one $(p,q)$-biclique and compute the probability of sampling it out. We return this probability by multiplying $B$ as the approximate count of $(p,q)$-bicliques (\Cref{Sampling:return}).
It is not hard to see that the time complexity is dominated by the cost of updating $S$, which requires enumerating all neighbors. The loop only lasts $O(p+q)$. Therefore, the time complexity is $O((p+q)\times \Delta)$, where $\Delta$ denotes the maximal degree in $G$.


\begin{algorithm}[t]
\caption{\texttt{CBS: Sampling}}\label{alg:Sampling}
\KwIn{$G$: a colored bipartite graph; $p,q$: two parameters; $Dp$: a $2D$ array; $B$: the total number of $(p,q)$-brooms in $G$.}
\KwOut{$ans$: estimate number of $(p,q)$-clique in $G$.}
\small
Sample the last edge $(u_{last},v_{last})$ in $E(G)$ following the weight distribution of brooms;\\ \label{sampling:samplelast}
$(u_{cur},v_{cur}) \leftarrow (u_{last},v_{last})$; \\ \label{sampling:initcur}

$U' \leftarrow \{ u_{last}\},\ V' \leftarrow \{ v_{last}\}$; \\ \label{sampling:inituv}

$ans \leftarrow 1$;\\ \label{sampling:initans}
\For{$i\leftarrow p+q-2$ \KwTo $1$}{ \label{sampling:for}
    Initialize set $S \leftarrow \emptyset$; \\ \label{sampling:initnextedge}
    \uIf{$Dir[i]=\texttt{Up}$}{ \label{sampling:ifleft}
        $S \leftarrow \{ (u_{cur},w) \ | \ w \in N(u_{cur},G),c(w) <c(v_{cur}),U'\subseteq N(w,G)\}$; \\ \label{sampling:leftupdate}
    }
    \Else{ \label{sampling:ifright}
        $S \leftarrow \{ (w,v_{cur}) \ | \ w \in N(v_{cur},G),c(w)<c(u_{cur}),V'\subseteq N(w,G)\}$; \\ \label{sampling:rightupdate}
    }
    $ans \leftarrow ans \times \frac{\sum_{(u,v) \in S} Dp[i][(u,v)]}{Dp[i+1][(u_{cur},v_{cur})]}$;\\ \label{sampling:updateans}
    \lIf{$ans = 0$}{  
        \textbf{break}  \label{sampling:break} 
    } 
    Sample the next edge $(u_{next},v_{next})$ in $S$ following the weight distribution of growing brooms;\\ \label{sampling:samplenext}
    $(u_{cur},v_{cur}) \leftarrow (u_{next},v_{next})$; \\ \label{sampling:updatecur}
    $U' \leftarrow U' \cup \{ u_{next}\}, V' \leftarrow V' \cup \{ v_{next}\}$; \\ \label{sampling:updateuv}
}
\Return $ans \times B$; \label{Sampling:return}
\end{algorithm}

\subsection{Overall Algorithm} 
We provide \Cref{alg:CBS} as a black box to call all three subroutines properly and output the estimated answer of the biclique count. We start by coloring (\Cref{CBS:init_Col}) and pattern counting (\Cref{CBS:init_Dp}). Then by the input sampling size parameter $T$, we repeatedly call \kw{Sampling}$(G,p,q,Dp,B)$ and aggregate the return to $\hat{C}$. The output approximate $(p,q)$-biclique count is the average of all $T$ attempts (\Cref{CBS:return}). Note that as a common trick, we will execute core-reduction for the original graph~\cite{yang2021p}: For any query $(p,q)$, we split the graph and reduce the query pair to $(p-1,q)$. Now, we give a theoretical analysis of \texttt{CBS} in the following.

\begin{algorithm}[t]
\caption{\texttt{CBS: Main}}\label{alg:CBS}
\KwIn{$G$: a bipartite graph; $p,q$: two parameters; $T$: sampling times.}
\KwOut{$\hat{C}$: estimate number of $(p,q)$-clique.}
\small
$G \leftarrow \ $ \kw{Coloring}$(G,p,q)$;\\ \label{CBS:init_Col}
$Dp,B \leftarrow \ $ \kw{CountingIndex}$(G,p,q)$; \\ \label{CBS:init_Dp}
$\hat{C} \leftarrow 0$; \\
\For{$i\leftarrow 1$ \KwTo $T$}{
    $\hat{C} \leftarrow \hat{C} \ + $ \kw{Sampling}$(G,p,q,Dp,B)$;
}
$\hat{C} \leftarrow \hat{C}/T$; \\

\Return $\hat{C}$; \label{CBS:return}
\end{algorithm}

\subsubsection{Unbiasedness.} We now show that the \texttt{CBS} method is unbiased. Let $(u_{cur}, v_{cur})$, $U'$, and $V'$ represent the current growing $(p, q)$-broom. For clarity, we introduce the following definitions:
\begin{itemize}[leftmargin=*] 
\item $\mathcal{F}_{U',V'}$: The total value to be multiplied into $ans$, defined as the product of the last $i = (p+q)-(|U'|+|V'|)$ fractions.
\item $\mathcal{B}_{U',V'}$: The number of $(p, q)$-brooms $H$ such that $U', V'$ are the maximal tuples of $U(H)$ and $V(H)$.
\item $\mathcal{C}_{U', V'}$: The number of $(p, q)$-bicliques $H'$ such that $U', V'$ are the maximal tuples of $U(H')$ and $V(H')$.
\end{itemize}
    
     After that, we prove the following lemma by mathematical induction.
    \begin{lemma} 
    \label{lemma:unbiased-f}
    For any growing $(p,q)$-broom $U',V'$,
    $
    \mathbb{E}\left[\mathcal{F}_{U',V'}\right] = \mathcal{C}_{U',V'} / \mathcal{B}_{U',V'}.
    $
    \end{lemma}
    
    \begin{proof}
        \textbf{(a)} When $|U'| + |V'|=p + q$, we have $\mathcal{B}_{U',V'}=1$. In this case, $\mathcal{F}_{U',V'}=\mathcal{C}_{U',V'}=1$ if $U',V'$ forms a $(p,q)$-biclique; otherwise $\mathcal{F}_{U',V'}=\mathcal{C}_{U',V'}=0$. Hence, 
           $\mathbb{E}\left[\mathcal{F}_{U',V'}\right] =\mathcal{C}_{U',V'} / \mathcal{B}_{U',V'}$ trivially holds in this case.

        \textbf{(b)} Suppose that the lemma holds for all growing $(p, q)$-brooms with $\ |U'| + |V'| = k + 1$ . For any broom with $|U'|+|V'|=k$, let $P = \frac{\sum_{(u,v) \in S} Dp[i][(u,v)]}{Dp[i+1][(u_{cur},v_{cur})]}.$ We then have
        \begin{align*}
            \mathbb{E}\left[\mathcal{F}_{U',V'}\right] = P \sum\limits_{(u,v) \in S} \frac{Dp[i][(u,v)]}{\sum\limits_{(u,v) \in S} Dp[i][(u,v)]} \mathbb{E}\left[\mathcal{F}_{U' \cup \{ u\},V' \cup \{ v\}}\right].
        \end{align*}
        Since $u_{cur}=u$ or $v_{cur}=v$ always holds, we have that $|U' \cup \{ u\}|$$+$$|V' \cup \{v\}|$$=|U|+|V|+1=k+1.$
        According to the definition of $Dp$ and $\mathcal{B}$, We can derive that $Dp[i][(u,v)]=\mathcal{B}_{U' \cup \{ u\},V' \cup \{ v\}}$, and thus
        \begin{align*}
            \mathbb{E}\left[\mathcal{F}_{U',V'}\right] &= \frac{\sum\limits_{(u,v) \in S} Dp[i][(u,v)]\frac{\mathcal{C}_{U' \cup \{ u\},V' \cup \{ v\}}}{\mathcal{B}_{U' \cup \{ u\},V' \cup \{ v\}}}}{Dp[i+1][(u_{cur},v_{cur})]} \\
            &= \frac{\sum\limits_{(u,v) \in S} \mathcal{C}_{U' \cup \{ u\},V' \cup \{ v\}}}{\mathcal{B}_{U',V'}} = \frac{\mathcal{C}_{U',V'}}{\mathcal{B}_{U',V'}},
        \end{align*}
    where the last equation is because $\mathcal{C}_{U' \cup \{ u\},V' \cup \{ v\}}=0$ for all pairs of $(u,v) \notin S$.
    \end{proof}
    \noindent 
    Then we can derive the following theorem:
    \begin{theorem}
    \label{theorem:unbiased-C_hat}
        Let $ans_i$ be the value of $ans$ in the $i$-th sampling. Let $\hat{C}=\frac{1}{T}\sum_{i=1}^{T} ans_i \times B$. Then $\hat{C}$ is an unbiased estimator of the number of $(p,q)$-bicliques in $G$.
    \end{theorem}
    \begin{proof}
        We first have the following induction by \Cref{lemma:unbiased-f}:
        \begin{align*}
             &\mathbb{E}\left[ans_i \times B\right] \\ &= B \times \sum\limits_{(u,v) \in E(G)} \frac{Dp[p+q-1][(u,v)]}{\sum Dp[p+q-1][(u,v)]} \mathbb{E}\left[\mathcal{F}_{\{u\},\{v\}}\right] \\
             &= \sum\limits_{(u,v) \in E(G)} Dp[p+q-1][(u,v)] \times \mathbb{E}\left[\mathcal{F}_{\{u\},\{v\}}\right] \\
             &= \sum\limits_{(u,v) \in E(G)} \mathcal{C}_{\{u\},\{v\}} = \mathcal{C}_{\emptyset,\emptyset}.  
        \end{align*}
        Then, we can derive that:
       $\mathbb{E}\left[\hat{C}\right]=\frac{1}{T}\sum_{i=1}^{T}\mathbb{E}\left[ans_i \times B\right]=\mathcal{C}_{\emptyset,\emptyset}.$
    
        Note that $\mathcal{C}_{\emptyset,\emptyset}$ is the number of $(p,q)$-bicliques in $G$ according to the definition. Therefore,  $\hat{C}$ is an unbiased estimator of the number of $(p,q)$-bicliques in $G$.
    \end{proof}
    Based on \Cref{theorem:unbiased-C_hat}, we have obtained an unbiased estimator $\hat{C}$ of the number of $(p,q)$-bicliques in \Cref{alg:CBS}.

\subsubsection{Error Analysis.}
we now analyze the estimation error of our sampling algorithms. Our analysis relies on the classic Hoeffding’s inequalities, which are shown below.

\begin{lemma}[Hoeffding’s inequality, \cite{chernoff1952measure,hoeffding1994probability}]
    \label{lemma:Hoeffding}
    For the random variables $X_i \in [0,M]$,$1\leq i \leq n$, we let $X = \sum_{i=1}^n X_i$. Then for $\epsilon>0$, we have
    \begin{align*}
        \Pr(X \ge (1+\epsilon)\mathbb{E}\left[X\right]) &\leq \exp(-\frac{2 \epsilon^2 \mathbb{E}\left[X\right]^2}{nZ^2}), \\
        \Pr(X \le (1-\epsilon)\mathbb{E}\left[X\right]) &\leq \exp(-\frac{2 \epsilon^2 \mathbb{E}\left[X\right]^2}{nZ^2}).
    \end{align*}
\end{lemma}

Based on \Cref{lemma:Hoeffding}, we can derive the estimation error of our algorithms as shown in the following theorem.

\begin{theorem}
    \label{theorem:estimation-error}
    Let $C,B$ be the number of $(p,q)$-bicliques, $(p,q)$-brooms, respectively. Let $\hat{C} = \frac{1}{T} \sum_{i=1}^{T} ans_i \times B$ denote the estimated number of $(p,q)$-bicliques, where $ans_i$ is the return result in the $i$-th sampling. Then, $\hat{C}$ is a $(1 + \epsilon)$ approximation of $C$ with probability $(1-\alpha)$ if $T \ge \frac{B^2}{2\epsilon^2 C^2}\ln(\frac{2}{\alpha})$. 
\end{theorem} 

\begin{proof}
    We can derive $ans_i\times B \in [0,B]$ since the values multiplied to $ans$ are all probabilities in $[0,1]$.
    And we have $\mathbb{E}\left[ \hat{C}T \right] =CT$ based on \Cref{theorem:unbiased-C_hat}.
    Given a positive value $\epsilon$, applying~\Cref{lemma:Hoeffding} by plugging in $n=T$, $X_i =ans_i \times B$, $X=\sum_{i=1}^nX_i=\hat{C}T$, we then have
    \begin{align*}
        \Pr(X \ge (1+\epsilon)\mathbb{E}\left[X\right]) &= \Pr(\hat{C}T\ge (1+\epsilon)CT) \leq \exp(-\frac{2 \epsilon^2(CT)^2}{TB^2}), \\
        \Pr(X \le (1-\epsilon)\mathbb{E}\left[X\right]) &= \Pr(\hat{C}T\le (1-\epsilon)CT) \leq \exp(-\frac{2 \epsilon^2(CT)^2}{TB^2}).
    \end{align*}
    Further, we have 
       $\Pr(\frac{|\hat{C}-C|}{C}\ge \epsilon) \leq 2\exp(-\frac{2 \epsilon^2C^2T}{B^2}).$
       
    Let $2\exp(-\frac{2 \epsilon^2C^2T}{B^2}) \le \alpha,$ we can derive that $T \ge \frac{B^2}{2\epsilon^2 C^2}\ln(\frac{2}{\alpha}).$
\end{proof}


\subsubsection{Time and Space Complexity.} 
The time complexity for each of the three parts in \texttt{CBS} are already discussed separately. We now give the following theorem as a summary.

\begin{theorem}
    \label{theorem:time}
The time complexity of \texttt{CBS} is $O(\kappa|E| + (p+q)|E| + T \cdot (p+q) \cdot \Delta)$, where $\Delta$ denotes the maximal degree in $G$. Specifically, to reach a $(1+\epsilon)$-approximation with probability $(1-\alpha)$, we require $T \ge \frac{B^2}{2\epsilon^2 C^2}\ln(\frac{2}{\alpha})$ where $C,B$ denotes the number of $(p,q)$-bicliques and $(p,q)$-brooms.
\end{theorem}

\begin{proof}
    The first two terms follow from the analysis in \Cref{sec:coloring} and \Cref{sec:patterncounting}. The third term follows from the fact that we need to conduct the sampling process (\Cref{alg:Sampling}) for $T$ times, and each time costs $(p+q)\cdot \Delta$. The approximation guarantee comes directly from \Cref{theorem:estimation-error}.
\end{proof}

We now provide the memory usage analysis as follows. We have also conducted an empirical experiment on the memory consumption in \Cref{sec:memory}. 
\begin{theorem}
\label{theorem:memory}
The space complexity of \texttt{CBS} is $O(|U|+|V|+(p+q)\cdot|E|)$.
\end{theorem}

\begin{proof}
In \Cref{alg:Coloring}, since we need to maintain $Cnt$ and $S_{next}$ which is at most $|U|+|V|$, the memory usage is $O(|U|+|V|)$. In \Cref{alg:CountingIndex}, the memory usage is dominated by the size of the DP table, thus $O((p+q)\cdot |E|)$. Algorithm 3 is for sampling a $(p,q)$-biclique, and $|S|$ is at most $\max \deg(u)$, thus the memory usage is $O(\max \deg(u))$. In all, the total space complexity is $O(|U|+|V|+(p+q)\cdot|E|)$.
\end{proof}

The \texttt{EP/Zz++} uses the heuristic framework to combine \texttt{EPivoter} (for exact counting on small portion of graphs) and \texttt{ZigZag++}. In the analysis of \texttt{ZigZag++}, its time and space complexity is $O((p+q)|E|+(p+q)^2 |V| + T \cdot \Delta \cdot (p+q)^2)$ and $O((p+q)|E|)$. Our memory usage is similar to each other. If we consider $T$ and $\kappa$ as a constant, our time complexity has some advantage as we don't have the term of $(p+q)^2$. However, it is obvious that both our algorithms' performance are highly determined by the size of $T$.

From \Cref{theorem:estimation-error}, the sample size $T$ is mainly determined by $(\frac{B}{C})^2$. That is to say, we need a larger sample size $T$ to ensure high accuracy with a larger $(\frac{B}{C})^2$, where $B$ is the number of $(p,q)$-brooms, $C$ is the number of $(p,q)$-bicliques. As shown in \Cref{sec:hyperparameters}, $(\frac{B}{C})^2$ is usually small comparing to the dominating term of \texttt{EP/Zz++} ($\frac{Z^2}{\rho^2}$)~\cite{ye2023efficient}. Therefore, our algorithm generally does not need a large $T$ to achieve good accuracy in real-world datasets. It can also be concluded that \texttt{CBS} guarantees an equivalent error ratio with significantly fewer sampling times and higher precision with the same sampling times compared to \texttt{EP/Zz++}.

\section{EVALUATION}
\label{sec:exp}
\begin{table}[t!]
\caption{Datasets used in experiments.}
\resizebox{1\linewidth}{!}{
\begin{tabular}{|l|rrrrr|} 
\hline
Graphs (Abbr.) & \multicolumn{1}{c}{\footnotesize Category} & \multicolumn{1}{c}{\footnotesize $|U|$} & \multicolumn{1}{c}{\footnotesize $|V|$} & \multicolumn{1}{c}{\footnotesize $|E|$} & \multicolumn{1}{c|}{\footnotesize $\max \kappa$} \\ 
\hline
\texttt{github (GH)} &   Authorship &	56,519 &	120,867 &	440,237 & 508\\
\texttt{StackOF (SO)} &  Rating &	545,195 &	96,678 &	1,301,942 & 290\\
\texttt{Twitter (Wut)} &	 Interaction &	175,214 &	530,418 &	1,890,661 & 2933\\
\texttt{IMDB (IMDB)} &   Affiliation &	685,568 &	186,414 &	2,715,604 & 158\\
\texttt{Actor2 (Actor2)} &	 Affiliation &	303,617 &	896,302 &	3,782,463 & 189\\
\texttt{Amazon (AR)} &	 Rating &	2,146,057 &	1,230,915 &	5,743,258 & 155\\
\texttt{DBLP (DBLP)} &	 Authorship &	1,953,085 &	5,624,219 &	12,282,059 & 126\\
\texttt{Epinions (ER)} &	 Rating &	120,492 &	755,760 &	13,668,320 & 13200\\
\texttt{Wikipedia-edits-de (DE)} &   Authorship &	1,025,084 &	5,910,432 &	129,885,939 & 118356\\
\hline
\end{tabular}
}
\label{tab:dataset}
\end{table}

\subsection{Experimental Setting}
\label{sec:exp:setup}
\paragraph{Datasets.} We use nine real datasets from different domains, which are available at SNAP~\cite{stanford_snap}, Laboratory of Web
Algorithmics~\cite{unimi_law}, and Konect ~\cite{konect}.
Table~\ref{tab:dataset} shows the statistics of these graphs.

\paragraph{Methods.}
We compare our method with the baselines from highly related works, summarized as follows:
\begin{itemize}[leftmargin=*]
  \item \texttt{BCList++}~\cite{yang2021p}: the biclique listing-based algorithm,
which is based on the Bron-Kerbosch algorithm \cite{bron1973finding}.
The key idea of \texttt{BCList++} is to iteratively enumerate all $(p,q)$-bicliques containing each vertex through a node expansion process. Specifically, it employs an ordering-based search paradigm, where for each vertex, only its higher-order neighbors are considered during enumeration. Additionally, a graph reduction technique is applied to reduce the search space before biclique counting.
    \item \texttt{EPivoter}~\cite{ye2023efficient}: the state-of-the-art algorithm for exact biclique counting, which relies on the edge-based pivot technique.
   In \texttt{EPivoter}, an edge-based search framework is introduced. Unlike \texttt{BCList++}, it iteratively selects edges from the candidate set (i.e., the edge set used for biclique expansion) to expand the current biclique.
    Specifically, during the enumeration process, in each branch, it first select one edge as the pivot edge, and based on this edge, vertices can be grouped into four disjoint groups.
  By storing the entire enumeration tree along with the four vertex sets, biclique counting can then be efficiently performed using combination counting.
    \item \texttt{EP/Zz++}~\cite{ye2023efficient}: the state-of-the-art algorithm for approximate counting.
    It first partitions the graph into two regions: a dense region (a subgraph containing high-degree vertices) and a sparse region (a subgraph containing low-degree vertices).
    For the sparse region, \texttt{EP/Zz++} utilizes \texttt{EPivoter} for exact counting, while for the dense region, it proposes a zigzag path-based sampling algorithm for approximate counting. 
  Specifically, it leverages the fact that a $(p,q)$-biclique must contain a fixed number of $(min\{p,q\})$-zigzag paths. Based on this property, the algorithm first counts the number of $h$-zigzag paths and then samples $T$ such paths, where $h=min\{p,q\}$. 
  Finally, it estimates the number of $h$-bicliques based on their proportion with the sampled paths.
    \item  \texttt{CBS}: our proposed approximation algorithm (\Cref{sec:algo}).
\end{itemize}

Notice that we compare \texttt{EP/Zz++}  instead of \texttt{EP/Zz}~\cite{ye2023efficient}, since the former one is a better version of the latter in terms of efficiency and accuracy. For completeness, we also conduct some initial experiment on \texttt{EP/ZZ}. Specifically, on the two largest datasets, \texttt{ER} and \texttt{DE}, \texttt{EP/Zz} failed to produce results within $10^5$ seconds for any values of $p$ and $q$ $(3 \leq p, q \leq 9)$. Thus, it is not applicable to scalable datasets in practice. Thus, we focus our evaluation primarily on \texttt{EP/Zz++} in our main experiments.

There is also another potential baseline introduced in \cite{ye2024fast}. The method takes over 1000 seconds to run on the DBLP dataset, while our algorithm requires only around 1 second. We have performed a primary initial evaluation based on the conclusions reported in the article. They claim to outperform \texttt{BCList++} and \texttt{EPivoter} by two orders of magnitude in runtime. However, both baselines fail to return results on the \texttt{ER} and \texttt{DE} datasets within $10^5$ seconds. Even based on their own estimation, their method remains significantly slower than ours in practice. As a result, we do not include it as a baseline in our empirical study.
We implement all the algorithms in C++ and run experiments on a machine having an Intel(R) Xeon(R) Platinum 8358 CPU @ 2.60GHz and 512GB of memory, with Ubuntu installed.

\paragraph{Parameter Settings.} In
our experiments, we following the existing work~\cite{ye2023efficient} setting $T$=$10^5$
as default values.
We define the estimation error as $\frac{|\hat{C}-C|}{C}$, where $C$ denotes the exact count of bicliques, and $\hat{C}$ represents its approximated value.
To ensure the reliability of our results, we run each approximation algorithm 10 times, with the reported error representing the average across these executions.
For some datasets, the exact count of bicliques cannot be computed (e.g., ER and DE), we evaluate the estimation error differently by using the average approximated biclique count from 10 runs as the exact count (i.e., $C$).

\subsection{Overall Comparison Results}
In this section, we compare \texttt{CBS} with three competitors (introduced in Section \ref{sec:exp:setup}), w.r.t. overall efficiency, accuracy, and the effect of $p$ and $q$, and the number of samples to demonstrate the superior of our algorithm.

\paragraph{1. Efficiency of All Algorithms.} 
\Cref{fig:main_time} depicts the average running time of all
the biclique counting algorithms on nine datasets for counting all $3 \leq p, q\leq9$ bicliques.
We make the following observations and analysis:
\begin{enumerate}[leftmargin=*]
    \item Our method \texttt{CBS} is up to two orders of magnitude faster than all competitors, this is mainly because our algorithm has a better theoretical guarantee.
    \item On almost all datasets, \texttt{CBS} is at least $10 \times$ faster than all methods, except DBLP dataset. On this dataset, \texttt{BCList++} achieves the best performance, as graph reduction significantly reduces $|U|$ and $|V|$, allowing it to perform efficiently without extra steps. Meanwhile, \texttt{CBS} still outperforms \texttt{EPivoter} and \texttt{EP/Zz++}.
    \item On the two largest datasets, ER and DE, the exact algorithms fail to terminate in $10^5$ seconds, while both approximate algorithms successfully complete the task. Our method, \texttt{CBS}, demonstrates at least 3 times faster performance compared to \texttt{EP/Zz++}.
\end{enumerate}


Note that besides \Cref{fig:main_time}, we choose to simply compare the sampling time with \texttt{EP/Zz++} in \Cref{figure:sampling_time_heatmap} and \Cref{figure:Timecurve}, which is for exposing the significant improvement in the sampling efficiency. 


\definecolor{c1}{RGB}{42,99,172} 
\definecolor{c2}{RGB}{255,88,93}
\definecolor{c3}{RGB}{208,167,39}
\definecolor{c4}{RGB}{119,71,64} 
\definecolor{c5}{RGB}{228,123,121} 
\definecolor{c6}{RGB}{175,171,172} 
\definecolor{c7}{RGB}{0,51,153}
\definecolor{c8}{RGB}{56,140,139} 

\definecolor{c9}{RGB}{0,0,0} 
\definecolor{c10}{RGB}{120,80,190} 

\definecolor{c11}{RGB}{255,204,0}
\definecolor{c12}{RGB}{128,128,128}
\definecolor{c13}{RGB}{98,148,96}
\definecolor{c14}{RGB}{184,168,207}
\definecolor{c15}{RGB}{253,207,158}
\definecolor{c16}{RGB}{182,118,108}
\definecolor{c17}{RGB}{175,171,172}

\definecolor{color1}{RGB}{0, 120, 190} 
\definecolor{color3}{RGB}{200, 180, 60} 

\definecolor{color2}{RGB}{220, 50, 50} 

\definecolor{p1}{RGB}{174,223,172} 
\definecolor{p2}{RGB}{224,175,107}  
\definecolor{p3}{RGB}{138,170,214}  
\definecolor{p4}{RGB}{222,117,123} 
\definecolor{p5}{RGB}{216,174,174} 
\definecolor{p6}{RGB}{163,137,214} 
\definecolor{p7}{RGB}{248,199,1} 
\definecolor{p8}{RGB}{205,205,205} 
\definecolor{p9}{RGB}{255,0,127} 

\definecolor{t1}{RGB}{148, 190, 146} 
\definecolor{t2}{RGB}{190, 149, 91} 
\definecolor{t3}{RGB}{0, 153, 255}
\definecolor{t4}{RGB}{255, 0, 51}

\definecolor{aya1}{RGB}{71, 85, 142}
\definecolor{aya2}{RGB}{105, 130, 185}
\definecolor{aya3}{RGB}{130, 45, 74}
\definecolor{aya4}{RGB}{174, 107, 129}

\definecolor{yoi1}{RGB}{205, 68, 50}
\definecolor{yoi2}{RGB}{235, 145, 99}

\definecolor{gy1}{RGB}{93, 116, 162}
\definecolor{gy2}{RGB}{196, 216, 242}

\definecolor{ci1}{RGB}{43, 48, 122}
\definecolor{ci2}{RGB}{119, 194, 243}
\definecolor{ci3}{RGB}{169, 111, 176}
\definecolor{ci4}{RGB}{216, 160, 199}

\definecolor{ni1}{RGB}{20, 54, 95}
\definecolor{ni2}{RGB}{118, 162, 185}
\definecolor{ni3}{RGB}{214, 79, 56}

\definecolor{hu1}{RGB}{199, 160, 133}
\definecolor{hu2}{RGB}{201, 71, 55}

\definecolor{kira1}{RGB}{105, 169, 78}
\definecolor{kira2}{RGB}{234, 199, 114}

\definecolor{yayi1}{RGB}{85, 59, 148}
\definecolor{yayi2}{RGB}{152, 114, 202}

\definecolor{ke1}{RGB}{69, 51, 112}
\definecolor{ke2}{RGB}{165, 151, 182}
 
\pgfplotstableread[col sep=comma]{Data/MainTime.csv}{\MainTime}
\pgfplotstableread[col sep=comma]{Data/DetailedCmp.csv}{\DetailedCmp}
\pgfplotstableread[col sep=comma]{Data/AR_Err.csv}{\ARErr}
\pgfplotstableread[col sep=comma]{Data/Wut_Err.csv}{\WutErr}
\pgfplotstableread[col sep=comma]{Data/AR_SampleTime.csv}{\ARSampleTime}
\pgfplotstableread[col sep=comma]{Data/ER_SampleTime.csv}{\ERSampleTime}

\pgfplotstableread[row sep=\\,col sep=&]{
    index & T & UncolTime & ColTime \\
    2 & $10^2$ & 153 & 371 \\
    4 & $10^3$ & 153 & 371 \\
    6 & $10^4$ & 155 & 372 \\
    8 & $10^5$ & 166 & 383 \\
    10 & $10^6$ & 274 & 491 \\
}\AblationERTime

\pgfplotstableread[row sep=\\,col sep=&]{
    index & T & UncolTime & ColTime & ACBS & DCBS\\
    2 & $10^2$ & 153 & 371 & 334 & 401\\
    4 & $10^3$ & 153 & 371 & 337 & 402\\
    6 & $10^4$ & 155 & 372 & 337 & 403\\
    8 & $10^5$ & 166 & 383 & 348 & 413\\
    10 & $10^6$ & 274 & 491 & 452 & 519\\
}\AblationERTimeUpd
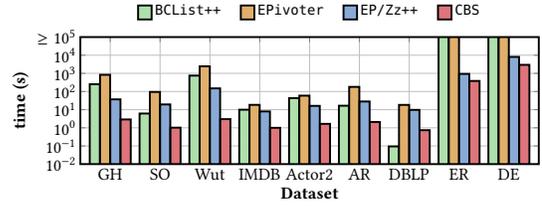
\begin{figure}[t!]
    \centering
       \begin{tikzpicture}[scale=0.45]
            \begin{axis}[
                    grid = major,
        		ybar=0.11pt,
        		bar width=0.3cm,
        		width=0.90\textwidth,
    			height=0.38\textwidth,
        		xlabel={\huge \bf Dataset}, 
        		xtick=data,	xticklabels={GH,SO,Wut,IMDB,Actor2,AR,DBLP,ER,DE},
                    legend style={at={(0.5,1.30)}, anchor=north,legend columns=-1,draw=none},
                    legend image code/.code={
                    \draw [#1] (0cm,-0.263cm) rectangle (0.4cm,0.15cm); },
        		xmin=0.8,xmax=19.2,
    			ymin=0.01, ymax = 100000,
                    ytick = {0.01, 0.1, 1, 10, 100, 1000, 10000, 100000},
    	        yticklabels = {$10^{-2}$, $10^{-1}$, $10^0$,$10^1$, $10^2$, $10^3$, $10^4$, $\geq 10^{5}$},
                    ymode = log,    
                    log basis y={2},
                    log origin=infty,
        		tick align=inside,
        		ticklabel style={font=\huge},
        	    every axis plot/.append style={line width = 1.6pt},
        		every axis/.append style={line width = 1.6pt},
                    ylabel={\textbf{\huge time (s)}},
        	]
        			\addplot[fill=p1] table[x=dataset,y=BCList++]{\MainTime};
        			\addplot[fill=p2] table[x=dataset,y=EPivoter]{\MainTime};
                        \addplot[fill=p3] table[x=dataset,y=EPZZ++]{\MainTime};
                        \addplot[fill=p4] table[x=dataset,y=OurMethod]{\MainTime};
                \legend{\huge {\tt BCList++ $\ $},\huge {\tt EPivoter $\ $}, \huge {\tt EP/Zz++ $\ $},\huge {\tt CBS}}
            \end{axis}
        \end{tikzpicture}
        \caption{Average runtime of different biclique counting algorithms for all $3 \leq p,q \leq 9$.}
    \label{fig:main_time}
\end{figure}

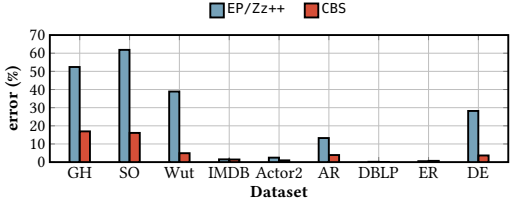
\begin{figure}[t!]
    \centering
       \begin{tikzpicture}[scale=0.45]
            \begin{axis}[
                    grid = major,
        		ybar=0.11pt,
        		bar width=0.3cm,
        		width=0.90\textwidth,
    			height=0.38\textwidth,
        		xlabel={\huge \bf Dataset}, 
        		xtick=data,	xticklabels={GH,SO,Wut,IMDB,Actor2,AR,DBLP,ER,DE},
                    legend style={at={(0.5,1.30)}, anchor=north,legend columns=-1,draw=none},
                    legend image code/.code={
                    \draw [#1] (0cm,-0.263cm) rectangle (0.4cm,0.15cm); },
        		xmin=0.8,xmax=19.2,
    			ymin=0, ymax = 70,
                    ytick = {0, 10, 20, 30, 40, 50, 60, 70},
    	        yticklabels = {$0$, $10$, $20$,$30$, $40$, $50$, $60$, $70$},
        		tick align=inside,
        		ticklabel style={font=\huge},
        	    every axis plot/.append style={line width = 1.6pt},
        		every axis/.append style={line width = 1.6pt},
                    ylabel={\textbf{\huge error (\%)}},
        	]
        			\addplot[fill=ni2] table[x=dataset,y=EPZZ++_error]{\DetailedCmp};
        			\addplot[fill=ni3] table[x=dataset,y=OurMethod_error]{\DetailedCmp};
                \legend{\huge {\tt EP/Zz++ $\ $},\huge {\tt CBS}}
            \end{axis}
        \end{tikzpicture}
        \caption{Average error of EP/Zz++ and CBS for all $3 \leq p,q \leq 9$.}
    \label{fig:error_cmp}
\end{figure}

\paragraph{2. Accuracy of All Algorithms.} Figure~\ref{fig:error_cmp} illustrates the average error rates of \texttt{CBS} and \texttt{EP/Zz++} across all datasets for biclique sizes ranging from 3 to 9 (i.e., $3 \leq p,q \leq 9$).
We can see that our algorithm  demonstrates up to a $8 \times$ reduction in error compared to \texttt{EP/Zz++}, thanks to our carefully designed coloring scheme and 
$(p,q)$-broom-based sampling technique.
Note that for the first three smaller datasets, \texttt{EP/Zz++} exhibits significant errors, with at least 38.9\% inaccuracy, rendering its results unreliable. In contrast, our method, \texttt{CBS}, maintains a maximum error rate of 16.9\% under the same number of sampling rounds. Our algorithm also consistently obtains an outstanding accuracy in larger datasets. A minor exception occurs in \texttt{ER} that our algorithm obtains a slightly worse error ratio in only one case $(p,q)=(4,9)$. To compensate for this, since our sampling time is more than 70 times faster compared to \texttt{EP/Zz++} (\Cref{figure:sampling_time_heatmap} (a)(b)), by simply increasing the sampling times, we can decrease the error ratio while still outperforming the baseline. Initially, we sampled $10^5$ times. By sampling $5\times 10^5$, $10^6$ and $7\times 10^6$, our error ratio decreases to $1.06\%$, $1.01\%$, $0.34\%$, respectively. These results all outperform \texttt{EP/Zz++} in both error ratio and sampling time.

Combining these observations with the performance results shown in Figure~\ref{fig:main_time}, we can conclude that our algorithm demonstrates an even greater advantage when considering the trade-off between accuracy and efficiency. This indicates that to achieve the same error rate, our algorithm would likely exhibit an even more substantial performance advantage over \texttt{EP/Zz++}.

\begin{figure}[t!]
    \centering
    \subfigure[ER, EP/Zz++]{
        \centering
        \includegraphics[width=0.45\linewidth]{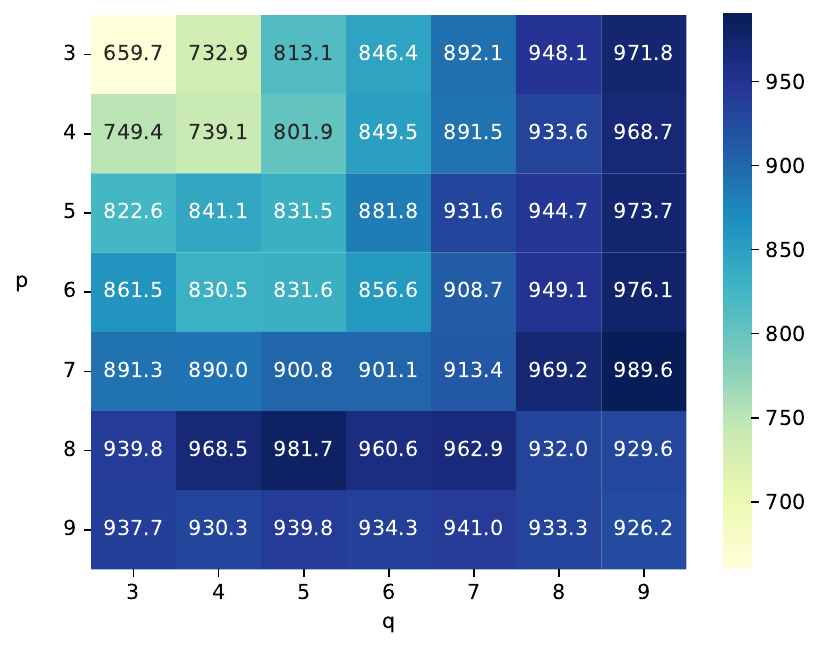}
    }
    \hspace{-0.1mm}
    \subfigure[ER, CBS]{
        \centering
        \includegraphics[width=0.45\linewidth]{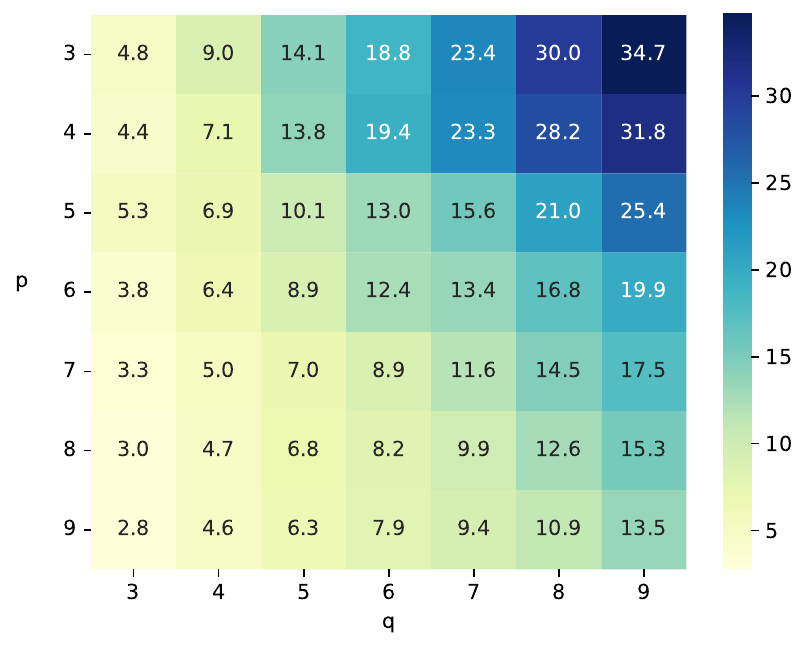}
    }
    \vspace{0.1cm}

    \subfigure[DE, EP/Zz++]{
        \centering
        \includegraphics[width=0.45\linewidth]{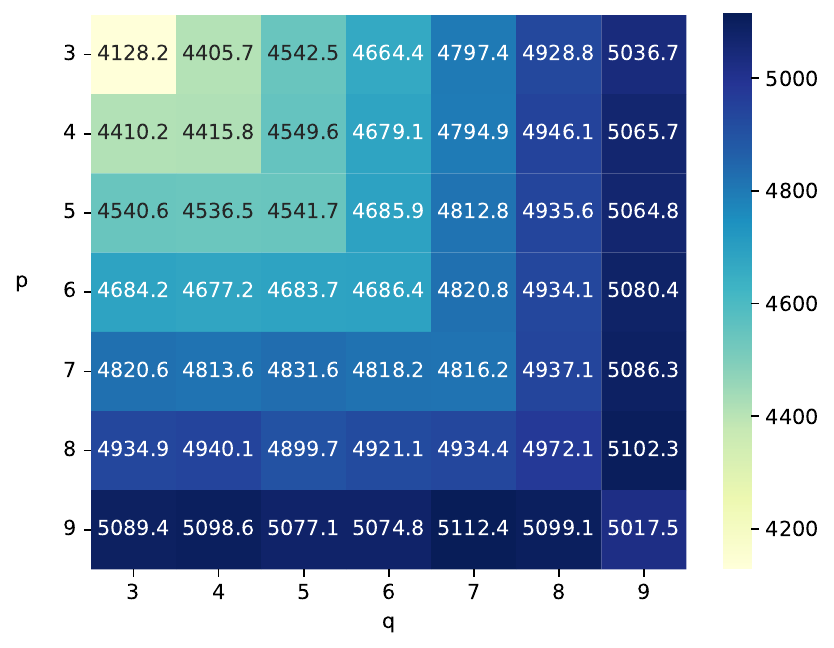}
    }
    \hspace{-0.1mm}
    \subfigure[DE, CBS]{
        \centering
        \includegraphics[width=0.45\linewidth]{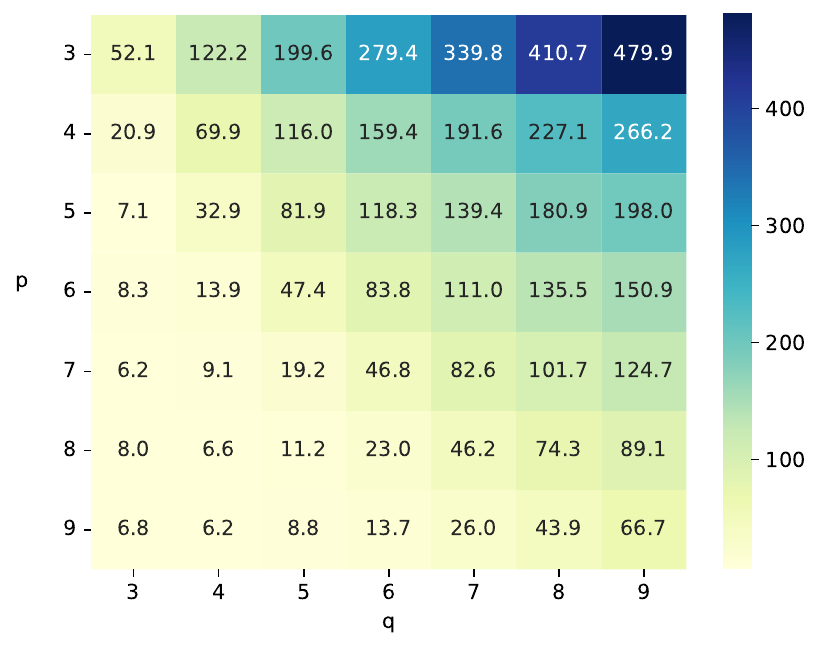}
    }
    \caption{The heat-map of sampling time of EP/Zz++ and CBS with varying $p$ and $q$ (s).}
    \label{figure:sampling_time_heatmap}
\end{figure}

\begin{figure}[t!]
    \centering
    \subfigure[Wut, EP/Zz++]{
        \centering
        \includegraphics[width=0.45\linewidth]{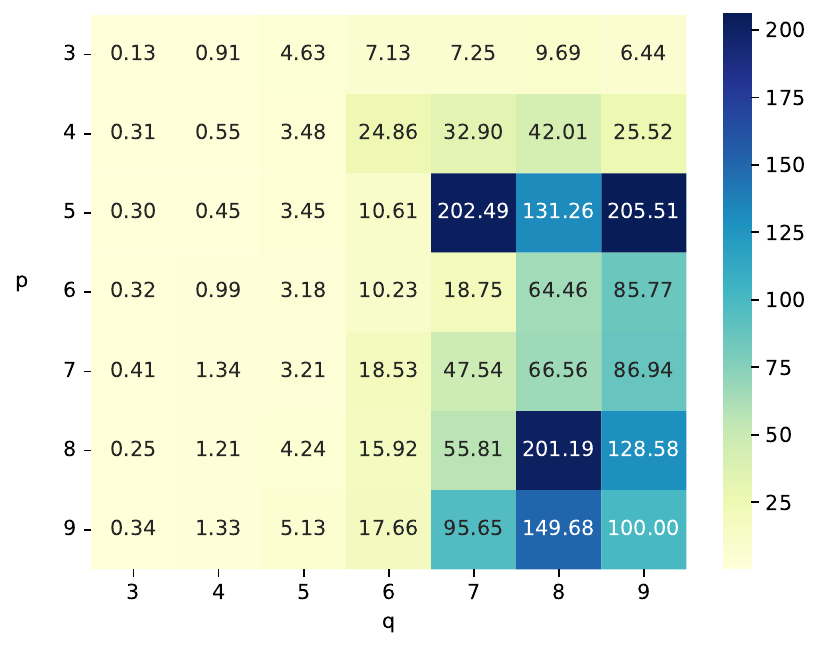}
    }
    \hspace{-0.1mm}
    \subfigure[Wut, CBS]{
        \centering
        \includegraphics[width=0.45\linewidth]{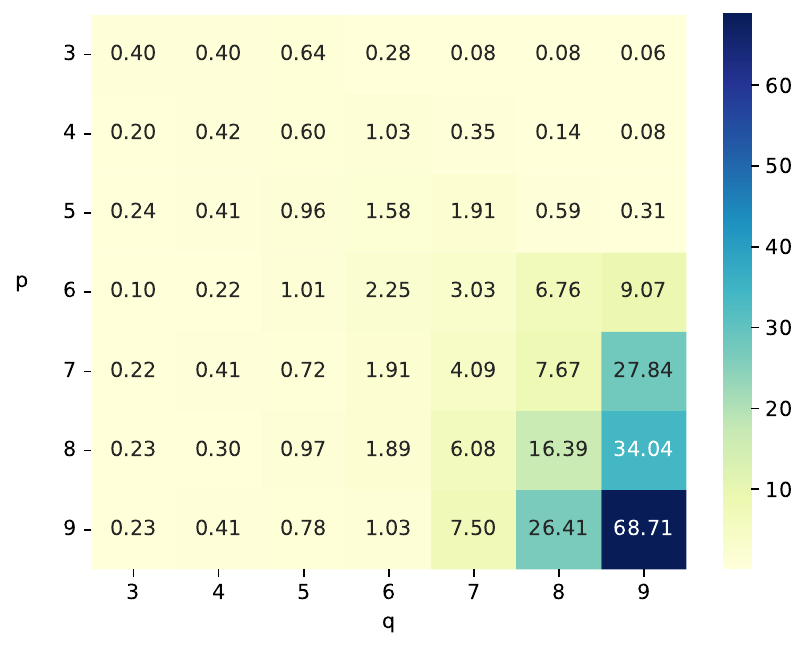}
    }
    \vspace{0.1cm}

    \subfigure[AR, EP/Zz++]{
        \centering
        \includegraphics[width=0.45\linewidth]{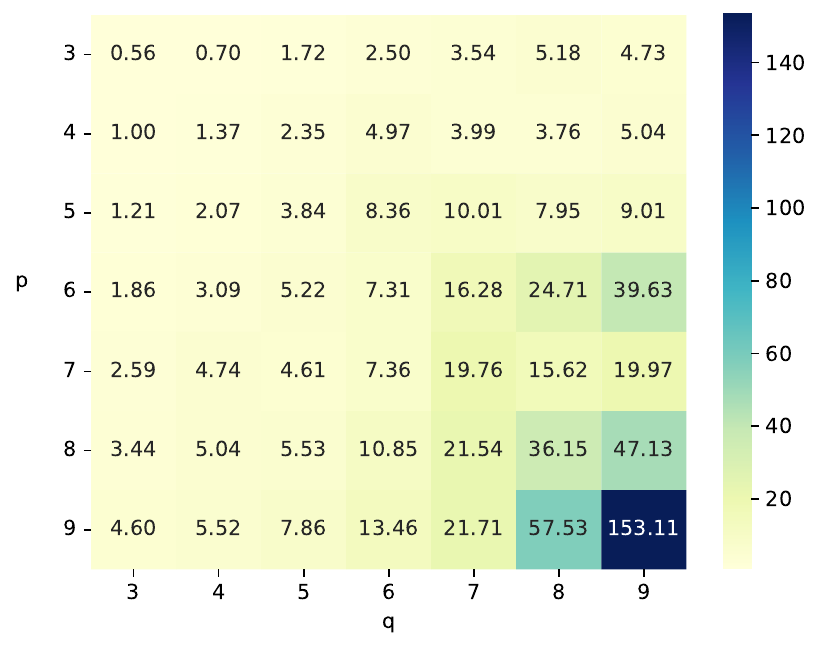}
    }
    \hspace{-0.1mm}
    \subfigure[AR, CBS]{
        \centering
        \includegraphics[width=0.45\linewidth]{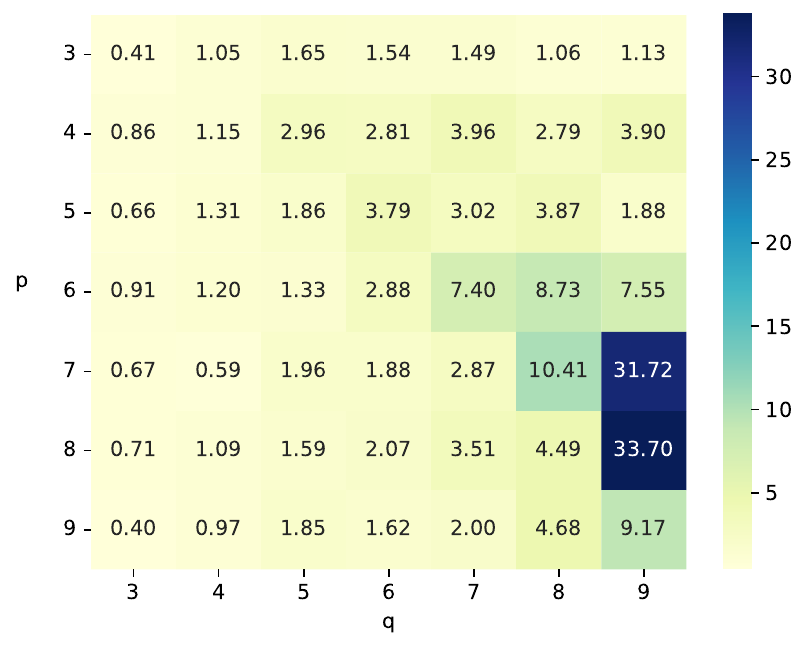}
    }
    \caption{The heat-map of estimation errors of EP/Zz++ and CBS with varying $p$ and $q$ (\%).}
    \label{figure:error_heatmap}
\end{figure}
\vspace{-1ex}
\paragraph{3. Effect of $p$ and $q$.}
Figures \ref{figure:sampling_time_heatmap} and \ref{figure:error_heatmap} illustrate the sampling time and error rate, respectively, of \texttt{CBS} and \texttt{EP/Zz++} across various $p$ and $q$ values.
In each figure, rows represent $p$
values and columns represent $q$ values. Each cell displays the sampling time (in Figure \ref{figure:sampling_time_heatmap}) or the estimation error (in Figure \ref{figure:error_heatmap}) for counting the corresponding $(p,q)$-bicliques.
Clearly, our method, \texttt{CBS}, outperforms \texttt{EP/Zz++} by up to two orders of magnitude in both sampling time and accuracy. For instance, with $p=5$ and $q=3$ on the ER dataset, \texttt{EP/Zz++} requires 822.6 seconds for sampling, whereas our algorithm completes this stage in just 5.3 seconds. 

In terms of accuracy, on the Wut dataset with $p=5$ and $q=7$, \texttt{EP/Zz++} has an estimation error of 202.49, while our algorithm achieves a significantly lower error of 1.91. This observation aligns with our analysis in Section 4. 
As a minor observation in Figure \ref{figure:sampling_time_heatmap}, the sampling time of \texttt{CBS} is more affected by q than by p while \texttt{EP/Zz++} does not exhibit this characteristic. This observation indeed demonstrates the difference between \texttt{Ep/Zz++} and our algorithm to some levels. Specifically, \texttt{EP/Zz++} needs to sample a zig-zag path with $h=\min(p,q)$ to approximate the number of $(p,q)$-bicliques, therefore its sampling time increases in a staircase manner w.r.t $\min(p,q)$. On the other hand, our algorithm samples brooms, which is different from \texttt{EP/Zz++}, this could be why we exhibit different characteristics. In our case, one of the possible reasons is that we start "growing" a broom from the left side during the sampling process.
\begin{figure}[t!]
    \centering
    \ref{Tcurve}\\
    \setlength{\abovecaptionskip}{-0.01cm}
    \subfigure[ER]{
    \begin{tikzpicture}[scale=0.4]
            \begin{axis}[
                mark size=5.0pt, 
                width=0.65\textwidth,
                height=0.31\textwidth,
                grid = major,
                xtick = {2,4,6,8,10,12},
                xticklabels = {$10^1$, $10^2$, $10^3$, $10^4$, $10^5$, $10^6$},
                ymin=0.001, ymax=10000,
                ytick = {0.001,0.01, 0.1, 1, 10, 100, 1000, 10000},
                yticklabels = {$10^{-3}$,$10^{-2}$, $10^{-1}$, $10^0$,$10^1$, $10^2$, $10^3$, $10^4$},
                ymode = log,
                log basis y={3},
                log origin=infty,
                xlabel={\Huge \bf T},
                ylabel={\Huge \bf sampling time(s)}, 
                ticklabel style={font=\Huge},
                every axis plot/.append style={line width = 2.5pt},
                every axis/.append style={line width = 2.5pt},
                ]
                \addplot[mark=square,color=ni2] table[x=dataset,y=EPZZ++] {\ERSampleTime};
                \addplot[mark=*,color=ni3] table[x=dataset,y=OurMethod]{\ERSampleTime};
            \end{axis}
    \end{tikzpicture}
    }
    \subfigure[AR]{
    \begin{tikzpicture}[scale=0.4]
            \begin{axis}[
                legend to name=Tcurve,
                legend style = {
				    legend columns=-1,
				    font=\small,
				    inner sep = 0pt,
				    draw=none,
			},
                mark size=5.0pt, 
                width=0.65\textwidth,
                height=0.31\textwidth,
                grid = major,
                xtick = {2,4,6,8,10,12},
                xticklabels = {$10^1$, $10^2$, $10^3$, $10^4$, $10^5$, $10^6$},
                ymin=0.001, ymax=1000,
                ytick = {0.001,0.01, 0.1, 1, 10, 100, 1000},
                yticklabels ={$10^{-3}$,$10^{-2}$, $10^{-1}$, $10^0$,$10^1$, $10^2$, $10^3$},
                ymode = log,
                log basis y={3},
                log origin=infty,
                xlabel={\Huge \bf T},
                ylabel={\Huge \bf sampling time(s)}, 
                ticklabel style={font=\Huge},
                every axis plot/.append style={line width = 2.5pt},
                every axis/.append style={line width = 2.5pt},
                ]
                \addplot[mark=square,color=ni2] table[x=dataset,y=EPZZ++] {\ARSampleTime};
                \addplot[mark=*,color=ni3] table[x=dataset,y=OurMethod]{\ARSampleTime};
                \legend{{\tt EP/Zz++},{\tt CBS}  }
            \end{axis}
    \end{tikzpicture}
    }
    \caption{Average sampling time of EP/Zz++ and CBS with varying $T$.}
    \label{figure:Timecurve}
\end{figure}
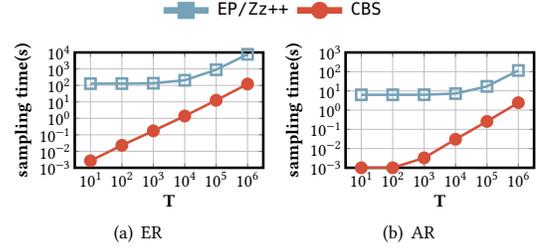

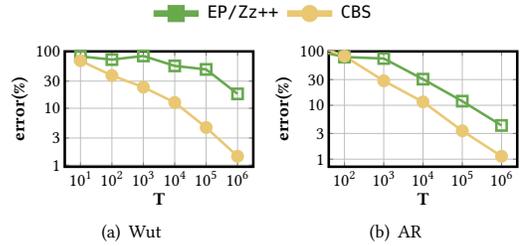
\begin{figure}[t!]
    \centering
    \vspace{1ex}
    \ref{Errorcurve}\\
    \setlength{\abovecaptionskip}{-0.01cm}
    \subfigure[Wut]{
    \begin{tikzpicture}[scale=0.4]
            \begin{axis}[
                mark size=5.0pt, 
                width=0.65\textwidth,
                height=0.31\textwidth,
                grid = major,
                xtick = {2,4,6,8,10,12},
                xticklabels = {$10^1$, $10^2$, $10^3$, $10^4$, $10^5$, $10^6$},
                ymin=0, ymax = 100,
                ytick = {1,3,10,30,100},
    	    yticklabels = {$1$, $3$, $10$, $30$,   $100$},
                ymode = log,
                xlabel={\Huge \bf T},
                ylabel={\Huge \bf error(\%)}, 
                ticklabel style={font=\Huge},
                every axis plot/.append style={line width = 2.5pt},
                every axis/.append style={line width = 2.5pt},
                ]
                \addplot[mark=square,color=kira1] table[x=dataset,y=EPZZ++] {\WutErr};
                \addplot[mark=*,color=kira2] table[x=dataset,y=OurMethod]{\WutErr};
            \end{axis}
    \end{tikzpicture}
    }
    \subfigure[AR]{
    \begin{tikzpicture}[scale=0.4]
            \begin{axis}[
                legend to name=Errorcurve,
                legend style = {
				    legend columns=-1,
				    font=\small,
				    inner sep = 0pt,
				    draw=none,
			},
                mark size=5.0pt, 
                width=0.65\textwidth,
                height=0.31\textwidth,
                grid = major,
                xtick = {2,4,6,8,10,12},
                xticklabels = {$10^1$, $10^2$, $10^3$, $10^4$, $10^5$, $10^6$},
                ymin=0, ymax = 100,
                ytick = {1,3,10,30,100},
    	    yticklabels = {$1$, $3$, $10$, $30$, $100$},
                ymode = log,
                xlabel={\Huge \bf T},
                ylabel={\Huge \bf error(\%)}, 
                ticklabel style={font=\Huge},
                every axis plot/.append style={line width = 2.5pt},
                every axis/.append style={line width = 2.5pt},
                ]
                \addplot[mark=square,color=kira1] table[x=dataset,y=EPZZ++] {\ARErr};
                \addplot[mark=*,color=kira2] table[x=dataset,y=OurMethod]{\ARErr};
                \legend{{\tt EP/Zz++},{\tt CBS}  }
            \end{axis}
    \end{tikzpicture}
    }
    \caption{Average error of EP/Zz++ and CBS with varying $T$.}
    \label{figure:T_curve}
\end{figure}

\paragraph{4. Effect of $T$.} We evaluate the effect of sample numbers on sampling time and error rate. The results on three datasets are shown in Figures \ref{figure:Timecurve} and \ref{figure:T_curve}. Based on these results, we observe the following: 
\vspace{-3.5ex}
\begin{enumerate}[leftmargin=*]
    \item For sampling time, our algorithm consistently outperforms \texttt{EP/Zz++}, particularly when the sample numbers are relatively low (e.g., from $10^1$ to $10^3$).
    \item Our algorithm consistently produces more accurate results than \texttt{EP/Zz++}, regardless of whether the sample size is high or low.
    \item Even with very few samples, our algorithm achieves acceptable solutions. For instance, on the AR dataset, \texttt{CBS} requires only $10^3$ samples to achieve a solution with a $30\%$ error rate, whereas \texttt{EP/Zz++} requires over $10^4$ samples.
\end{enumerate}


\section{Detailed Analysis}
\label{exp:detail}
In this section, to provide a better understanding of our algorithm, besides the primary evaluation in \Cref{sec:exp}, we now extensively evaluate and analyze \texttt{CBS} from different angles:

\textbf{$\bullet$} An ablation study on the effect of our coloring technique, showing that it can significantly improve accuracy without highly impacting sampling time.

\textbf{$\bullet$} An ablation study on the coloring technique conducted using different ordering strategies.

\textbf{$\bullet$} A study on the time cost of each step in \texttt{CBS}, including the initialization and sampling time, compared to the baseline \texttt{EP/Zz++}.

\textbf{$\bullet$} A study on the statistical of the hyper-parameters that determines the error ratio, which provides evidence that our algorithm \texttt{CBS} consistently requires fewer samples than \texttt{EP/Zz++}.

\textbf{$\bullet$} A study of the memory consumption of our algorithm.

\textbf{$\bullet$} An additional scalability experiment conducted by randomly sampling different percentages of vertices in the same graph.

\textbf{$\bullet$} An additional efficiency experiment with small $p,q$ aligning with the empirical study in \cite{ye2023efficient}.
\subsection{Ablation Study on the Coloring Technique} 
\label{exp:color}
We conduct an ablation study on the effect of our coloring technique. 
To begin with, in order to evaluate the effect of our coloring technique, we design a new variant of \texttt{CBS} by removing the vertex coloring step, denoted by \texttt{BS}. Then, note that in our proposed coloring methods, we decide to traverse the neighbors in a random order manner. A question is raised that if we follow different order of neighbors, would the performance be significantly affect? To answer this question, we design three more variants of \texttt{CBS} by using two different ordering strategies in coloring: the ascending order and descending order of degrees, denoting as \texttt{ACBS} and \texttt{DCBS}. 
 
To compare the inaccuracy, we first run these variants on the two datasets AR and Wut with varying parameters of $T$, and report the results in \Cref{fig:ablation-study2}. Compared to \texttt{BS} (without the coloring), \texttt{CBS} generally achieves lower error with the same number of samples, aligning with our previous analysis. Although the coloring technique initially slightly increases the sampling time, the gap diminishes as the number of sampling iterations increases. This demonstrates that in high-accuracy scenarios, our coloring technique can significantly enhance accuracy without significantly impacting the sampling time. Compared to \texttt{ACBS} and \texttt{DCBS}, there is no significant difference in error in the three strategies. This confirms that our algorithm is robust to the choice of coloring order.

In \Cref{fig:ablation-study3}, we run all four algorithms on the dataset ER with varying values of $T$ to compare their efficiency. We can see that without the coloring (\texttt{BS}), the total runtime does not decrease much. On the other hand, the average runtime of \texttt{CBS}, \texttt{ACBS}, and \texttt{DCBS} is similar. This again confirms the robustness of our coloring method. Combining with the accuracy result, we can see that the coloring method is effective, as with a slight increase in runtime, the accuracy is enhanced. Overall, our ablation study demonstrates that the coloring process is vital in counting bicliques.

\begin{figure}[t!]
    \centering
    \begin{minipage}[t]{0.9\linewidth}  
        \centering
        \resizebox{1.0\linewidth}{!}{
            \begin{tabular}{c|cccc|cccc}
                \bottomrule  
                \multirow{2}{*}{$T$} & \multicolumn{4}{c|}{AR}  & \multicolumn{4}{c}{Wut} \\ 
                \cline{2-9} 
               & \tt BS  & \tt CBS & \tt ACBS & \tt DCBS & \tt BS  & \tt CBS & \tt ACBS & \tt DCBS \\
                 \midrule
                 $10^2$ & 72.75 & 81.63 & 98.50 & 78.99 & 60.60 & 37.70 & 87.07 & 33.91 \\
                 $10^3$ & 43.33 & 28.47 & 32.44 & 32.43 & 22.03 & 23.50 & 24.48 & 21.16 \\
                 $10^4$ & 13.08 & 11.54 & 11.55 & 10.62 & 13.51 & 12.64 & 13.41 & 12.89\\
                 $10^5$ & 4.37 & 3.35 & 3.92 & 3.40 & 5.11 & 4.59 & 4.44 & 3.98\\
                 $10^6$ & 1.35 & 1.13 & 1.19 & 1.06 & 1.71 & 1.43 & 1.67 & 1.42\\
                \bottomrule      
            \end{tabular}
        }
    \end{minipage}
    \caption{Average error of different coloring techniques 
 in AR and Wut with varying $T$.}
    \label{fig:ablation-study2}
\end{figure}

\begin{figure}[t!]
    \centering
    \begin{minipage}{0.6\linewidth}  
        \centering
        \ref{Colorcurve}\\
        \setlength{\abovecaptionskip}{-0.01cm}
        \centering
        \begin{tikzpicture}[scale=0.6]
            \begin{axis}[
                legend to name = Colorcurve,
                legend style = {
				    legend columns=-1,
				    font=\small,
				    inner sep = 0pt,
				    draw=none,
			},
                mark size=4.0pt, 
                width=1.5\textwidth,
                height=1.0\textwidth,
                grid = major,
                xtick = {2,4,6,8,10},
                xticklabels = {$10^2$, $10^3$, $10^4$, $10^5$, $10^6$},
                ymin=50, ymax = 800,
                ytick = {50,100,200,400,800},
    	    yticklabels = {$50$,$100$,$200$,$400$,$800$},
                ymode = log,
                xlabel={\huge \bf $T$},
                ylabel={\huge \bf total time (s)}, 
                ticklabel style={font=\huge},
                every axis plot/.append style={line width = 1.5pt},
                every axis/.append style={line width = 1.5pt}
                ]
                \addplot[mark=square,color=ni2] table[x=index,y=UncolTime] {\AblationERTimeUpd};
                \addplot[mark=*,color=ni3] table[x=index,y=ColTime]{\AblationERTimeUpd};
                \addplot[mark=triangle,color=kira1] table[x=index,y=ACBS]{\AblationERTimeUpd};
                \addplot[mark=o,color=kira2] table[x=index,y=DCBS]{\AblationERTimeUpd};
                \legend{{\tt BS},{\tt CBS},{\tt ACBS},{\tt DCBS}};
            \end{axis}
        \end{tikzpicture}
    \end{minipage}
    \caption{Average runtime of different coloring techniques 
 in ER with varying $T$.}
    \label{fig:ablation-study3}
\end{figure}

\subsection{Time Cost of Different Stages}
\label{exp:stage}
The two approximation algorithms, \texttt{CBS} and \texttt{EP/Zz++}, both require an initialization stage to precompute some auxiliary data for sampling. 
Specifically, in \texttt{CBS}, we need to assign a color number for each vertex (i.e., coloring) and count the number of $(p,q)$-brooms in bipartite graph, while in \texttt{EP/Zz++}, it needs to count the number of $h$-zigzag paths bipartite graph,  where $h$ = $\min\{p,q\}$.
In Figure ~\ref{fig:init_sample_cmp}, we report the initializing and sampling times for these algorithms across all datasets.
We observe that on more than half of the datasets, our algorithm requires less initializing and sampling time, which indicates that the number of $(p,q)$-brooms in the bipartite graph is typically less than $h$-zigzag paths (using in \texttt{EP/Zz++}). In terms of sampling time, our algorithm is up to two orders of magnitude faster than \texttt{EP/Zz++}. While on two datasets, GH and ER, \texttt{EP/Zz++} is 10 times faster in initializing, \texttt{CBS} achieves 100 times lower sampling time, making the total time (i.e., initializing plus sampling) of our algorithm still highly efficient.

\begin{figure}[t!]
    \centering
       \begin{tikzpicture}[scale=0.45]
            \begin{axis}[
                    grid = major,
        		ybar=0.11pt,
        		bar width=0.3cm,
        		width=0.90\textwidth,
    			height=0.38\textwidth,
        		xlabel={\huge \bf Dataset}, 
        		xtick=data,	xticklabels={GH,SO,Wut,IMDB,Actor2,AR,DBLP,ER,DE},
                    legend style={at={(0.5,1.30)}, anchor=north,legend columns=-1,draw=none},
                    legend image code/.code={
                    \draw [#1] (0cm,-0.263cm) rectangle (0.4cm,0.15cm); },
        		xmin=0.8,xmax=19.2,
    			ymin=0.01, ymax = 100000,
                    ytick = {0.01, 0.1, 1, 10, 100, 1000, 10000, 100000},
    	        yticklabels = {$10^{-2}$, $10^{-1}$, $10^0$,$10^1$, $10^2$, $10^3$, $10^4$, $\geq 10^{5}$},
                    ymode = log,    
                    log basis y={2},
                    log origin=infty,
        		tick align=inside,
        		ticklabel style={font=\huge},
        	    every axis plot/.append style={line width = 1.6pt},
        		every axis/.append style={line width = 1.6pt},
                    ylabel={\textbf{\huge time (s)}},
        	]
        			\addplot[fill=gy1] table[x=dataset,y=EPZZ++_init]{\DetailedCmp};
        			\addplot[fill=gy2] table[x=dataset,y=OurMethod_init]{\DetailedCmp};
                        \addplot[fill=ci3] table[x=dataset,y=EPZZ++_sample]{\DetailedCmp};
                        \addplot[fill=ci4] table[x=dataset,y=OurMethod_sample]{\DetailedCmp};
                \legend{\huge {\tt EP/Zz++\_init $\ $},\huge {\tt CBS\_init $\ $}, \huge {\tt EP/Zz++\_sample $\ $},\huge {\tt CBS\_sample}}
            \end{axis}
        \end{tikzpicture}
        \caption{Average initializing and sampling time of EP/Zz++ and CBS for all $3 \leq p,q \leq 9$ ($T=10^5$).}
    \label{fig:init_sample_cmp}
\end{figure}
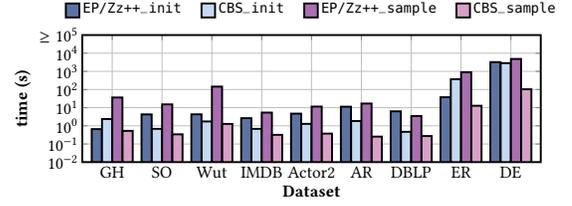

\subsection{Statistical of the Hyper-parameters}
\label{sec:hyperparameters}
Recall that in \texttt{CBS} and \texttt{EP/Zz++}, when $\epsilon$ is fixed, their required sample sizes are proportional to $\frac{B^2}{C^2}$ and $\frac{Z^2}{\rho^2}$, respectively. Therefore, the size of these two values in the graph data determines the efficiency of the two algorithms. We conduct an experiment to compare these two values in real-world data. Specifically, as shown in \Cref{tab:dataset_stat}, we report the values of $\frac{B^2}{C^2}$ and $\frac{Z^2}{\rho^2}$ on the Amazon dataset for varying $p$ and $q$. Similar trends are observed across other datasets.
Based on \Cref{tab:dataset_stat}, the value of $\frac{B^2}{C^2}$ is consistently smaller than $\frac{Z^2}{\rho^2}$ for all pairs of $p$ and $q$, up to 100$\times$. Therefore, our algorithm \texttt{CBS} consistently requires fewer samples than \texttt{EP/Zz++}. This also explains why \texttt{CBS} achieves the same or even lower estimation error with a smaller sample size. For instance, \texttt{CBS} requires up to 88$\times$ fewer samples than \texttt{EP/Zz++} on $p$ = 5 and $q$ = 9, showing its efficiency in reducing sampling overhead while maintaining accuracy.

\begin{table}[t!]
\caption{The value of $\frac{B^2}{C^2}$ and $\frac{Z^2}{\rho^2}$ with varying $p,q$ (Amazon).}
\resizebox{0.45\linewidth}{!}{
\begin{tabular}{l|rr} 
\bottomrule 
$(p,q)$ & \multicolumn{1}{c}{\footnotesize $\frac{B^2}{C^2}$} & \multicolumn{1}{c}{\footnotesize $\frac{Z^2}{\rho^2}$} \\ 
\midrule
(3,4)&	1.40E+03&	1.79E+03\\
(3,5)&	3.97E+03&	1.69E+04\\
(3,9)&	4.55E+04&	2.85E+05\\
(4,5)&	5.41E+04&	5.98E+04\\
(4,8)&	9.27E+04&	2.48E+05\\
(5,3)&	3.96E+02&	6.17E+03\\
(5,6)&	4.04E+05&	1.20E+06\\
(5,9)&	3.77E+04&	3.34E+06\\
(6,4)&	5.73E+03&	5.80E+04\\
(6,7)&	6.42E+06&	5.32E+07\\
(7,4)&	8.07E+03&	1.54E+05\\
(7,7)&	7.15E+06&	2.44E+07\\
(8,4)&	1.36E+04&	3.16E+05\\
(8,8)&	4.74E+07&	1.04E+09\\
(9,4)&	2.36E+04&	5.32E+05\\
(9,9)&	2.92E+08&	8.01E+10\\
\bottomrule 
\end{tabular}
}
\label{tab:dataset_stat}
\end{table}

\vspace{-2ex}
\subsection{Memory Consumption}
\label{sec:memory}
The theoretical guarantee of the memory usage is given in \Cref{theorem:memory}. We now present the actual memory consumption on each dataset in \Cref{tab:memory_consumption} (in Megabytes). Specifically, we report the average memory usage on each dataset for $3 \leq p,q \leq 9$. As we can observe, the memory usage is similar to all other baselines for most datasets. However, our memory consumption increases for the two large datasets. This is in line with our analysis, as it is dominated by $(p+q)|E|$. Our algorithm requires to store the intermediate result for retrieving the broom structure. In result, the memory is mainly for storing the dynamic programming table.
\begin{table}[t!]
\centering
\caption {Memory consumption of different counting algorithms (in MBs).}
\resizebox{0.8\linewidth}{!}{
\begin{tabular}{l|ccccc}
\bottomrule 
dataset & \texttt{BCList++} & \texttt{EPivoter} & \texttt{EP/Zz++} & \texttt{CBS} \\
\midrule
GH & 21.3 & 880.8 & 21.3 & 462.1 \\
SO & 40.9 & 901.0 & 73.7 & 120.8 \\
Wut & 52.2 & 911.0 & 69.8 & 312.8 \\
IMDB & 66.5 & 926.0 & 106.1 & 164.2 \\
Actor2 & 88.3 & 946.9 & 103.1 & 268.5 \\
AR & 152.7 & 1004.5 & 295.8 & 461.0 \\
DBLP & 319.5 & 1166.3 & 444.5 & 718.2 \\
ER & X & X & 402.7 & 46600.6 \\
DE & X & X & 1744.8 & 195153.0 \\
\bottomrule 
\end{tabular}
}
\label{tab:memory_consumption}
\end{table}

\vspace{-3px}
\subsection{Scalability Experiment}
\label{exp:scala}
We consider counting bicliques approximately, mainly because all exact counting methods are highly inefficient in large graph data. Therefore, it is essential to verify the scalability of our algorithm thoroughly. Besides the main experiment in \Cref{sec:exp}, we now conduct an additional graph scalability experiment that evaluates performance under dynamic modifications, which is done by randomly removing a percentage of edges. Specifically, we construct several graphs from DE as follows: 
\begin{enumerate}[leftmargin=*]
    \item Randomly sample $60\%$, $80\%$, $100\%$ of the vertices on both sides and keep the induced graph.
    \item Execute our algorithm on $p=6,q=8,T=10^5$.
    \item Repeat (1) and (2) 10 times, compute the average runtime for initialization and sampling.
\end{enumerate}

The three constructed graphs are DE60, DE80, and DE100. Shown in \Cref{tab:Scalability}, we can observe that as the size of the graph increases, the sampling time also increases linearly while the error ratio is within a certain range. This shows that \texttt{CBS}'s scalability, which is expected to be particularly effective on large datasets.

\begin{table}[t]
\centering
\caption {Scalability experiment.}
\resizebox{0.6\linewidth}{!}{
\begin{tabular}{l|ccc}
\bottomrule
& init time(s) & sample time(s) & error(\%) \\
\midrule
DE60\% & 298.6 & 40.8 & 3.76 \\
DE80\% & 1216.7 & 95.3 & 2.49 \\
DE100\% & 2478.1 & 135.5 & 3.19 \\
\bottomrule 
\end{tabular}
}
\label{tab:Scalability}
\end{table}

\vspace{2ex}
\subsection{Efficiency Comparison on Small $p,q$}
\label{exp:smallpq}
One of our major baselines, \texttt{EPivoter} is designed for exact counting on small $p,q$ (e.g., $p,q\leq5$). However, our major comparison experiment is conducted for all combinations of $p,q \leq 9$. To bridge the gap, we provide an additional experiment using the same parameter settings as \texttt{EPivoter}, i.e., $p, q \leq 5$.
Shown in \Cref{tab:small_cmp}, for ER and DE, \texttt{EPivoter} can not terminate in $ 10^5$ seconds. For other datasets, we observe that our method maintains its performance advantage even when $p$ and $q$ are small, highlighting its robustness to parameter settings. These additional results confirm that our method performs consistently well regardless of the parameter choices.

\begin{table}[t]
\centering
\caption{Average runtime (s) of EPivoter and CBS for all $3 \leq p,q \leq 5$.}
\resizebox{0.5\linewidth}{!}{
\begin{tabular}{l|cc}
\toprule
\textbf{Dataset} & \textbf{EPivoter} & \textbf{CBS} \\
\midrule
GH               & 805.4             & 3.4          \\
SO               & 95.0              & 1.4          \\
Wut              & 2429.4            & 4.3          \\
IMDB             & 18.4              & 1.6          \\
Actor2           & 53.4              & 2.8          \\
AR               & 180.6             & 3.7          \\
DBLP             & 18.5              & 1.9          \\
ER               & INF               & 364.4        \\
DE               & INF               & 3782.1       \\
\bottomrule
\end{tabular}
}
\label{tab:small_cmp}
\end{table}
\vspace{-1ex}
\section{CONCLUSION}
\label{sec:conc}

In this paper, we tackle the problem of $(p,q)$-biclique counting in large-scale bipartite graphs, crucial for applications such as recommendation systems and cohesive subgraph analysis. To address scalability and accuracy issues in existing methods, we propose a novel sampling-based algorithm that leverages \emph{$(p,q)$-brooms}, special spanning trees within $(p,q)$-bicliques. By utilizing the graph coloring method and dynamic programming, our method efficiently approximates $(p,q)$-biclique counts with unbiased estimates and provable error guarantees. Experimental results on nine real-world datasets show that our approach outperforms state-of-the-art methods, achieving up to 8$\times$ error reduction and 50$\times$ speed-up. 
Interesting future work includes extending our method to dynamic bipartite graphs with evolving structures and exploring its application to counting motifs/cliques in heterogeneous information networks.

\vspace{-1ex}
\begin{acks}
Chenhao Ma was partially supported by NSFC under Grant 62302421, Basic and Applied Basic Research Fund in Guangdong Province under Grant 2025A1515010439, and the Guangdong Provincial Key Laboratory of Big Data Computing, The Chinese University of Hong Kong, Shenzhen. Weinuo Li and Can Wang were partially supported by the National Natural Science Foundation of China (No. 62476244), the Starry Night Science Fund of Zhejiang University Shanghai Institute for Advanced Study, China (Grant No: SN-ZJU-SIAS-001). Parts of this work was done while Jingbang was still at University of Waterloo.
\end{acks}
\clearpage
\bibliographystyle{ACM-Reference-Format}
\balance
\bibliography{References}
\newpage
\clearpage

\end{document}